\documentclass[preprint,8pt,authoryear]{elsarticle}
\usepackage[a4paper, margin=1in]{geometry} 
\journal{Transportation Research Part C: Emerging Technologies}

\usepackage{graphicx}
\usepackage{eqndefns-left}
\usepackage{multirow}
\usepackage{hhline}
\usepackage{pdflscape}

\usepackage[utf8]{inputenc} 
\usepackage[T1]{fontenc}    
\usepackage{hyperref}       
\usepackage{url}            
\usepackage{booktabs}       
\usepackage{amsfonts}       
\usepackage{nicefrac}       
\usepackage{microtype}      
\usepackage{cleveref}       
\usepackage{graphicx}
\usepackage{doi}
\usepackage{fancyhdr}
\usepackage{lineno}

\usepackage{xcolor,soul}

\usepackage{todonotes}
\usepackage[small, margin=1cm]{caption}
\usepackage{subcaption}

\usepackage{appendix}
\usepackage{color}
\definecolor{strcolor}{rgb}{0.6, 0.2, 0.6}
\definecolor{commentcolor}{rgb}{0.3125, 0.5, 0.3125}
\definecolor{keycol}{rgb}{0, 0, 1}

\usepackage{enumitem}

\usepackage{amsthm}
\newtheorem{thm}{Result}

\usepackage{bbm}

\newcommand {\bea}{\begin{eqnarray}}
	\newcommand {\eea}{\end{eqnarray}}

\def\blot{\quad \mbox{$\vcenter{ \vbox{ \hrule height.4pt
				\hbox{\vrule width.4pt height.9ex \kern.9ex \vrule width.4pt}
				\hrule height.4pt}}$}}

\usepackage{natbib}
\bibpunct[, ]{(}{)}{,}{a}{}{,}%
\gdef\AQ#1{}
\gdef\CQ#1{}

\hypersetup{
pdftitle={Semi-on-Demand Hybrid Transit Route Design with Shared Autonomous Mobility Services},
pdfsubject={q-bio.NC, q-bio.QM},
pdfauthor={Max T.M.~Ng, Florian~Dandl, Hani S.~Mahmassani, Klaus~Bogenberger},
pdfkeywords={semi-on-demand, transit feeder, shared autonomous mobility service, transit design, flexible route},
}

\begin{document}

\begin{frontmatter}

\title{Semi-on-Demand Hybrid Transit Route Design with Shared Autonomous Mobility Services}

\newcommand{\orcidicon}{\includegraphics[scale=0.06]{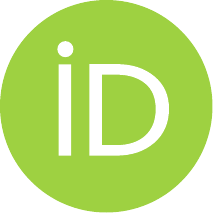}}
\author[1]{\orcidicon~\href{https://orcid.org/0000-0002-1252-320X}{Max T.M. Ng}}

\author[2]{\orcidicon~\href{https://orcid.org/0000-0003-3706-9725}{Florian Dandl}\corref{cor1}}
\cortext[cor1]{Corresponding author: \texttt{florian.dandl@tum.de}}

\author[1]{\orcidicon~\href{https://orcid.org/0000-0002-8443-8928}{Hani S. Mahmassani}\textsuperscript{\dag}}

\author[2]{\orcidicon~\href{https://orcid.org/0000-0003-3868-9571}{Klaus Bogenberger}}

\address[1]{Northwestern University Transportation Center, 600 Foster St, Evanston, IL 60208, USA}
\address[2]{Chair of Traffic Engineering and Control, Technical University of Munich, Arcisstraße 21, 80333 Munich, Germany}

\begin{abstract}
Shared Autonomous Vehicles (SAVs) enable transit agencies to design more agile and responsive services at lower operating costs. This study designs and evaluates a semi-on-demand hybrid route directional service in the public transit network, offering on-demand flexible route service in low-density areas and fixed route service in higher-density areas. We develop analytically tractable cost expressions that capture access, waiting, and riding costs for users, and distance-based operating and time-based vehicle costs for operators. Two formulations are presented for strategic and tactical decisions in flexible route portion, fleet size, headway, and vehicle size optimization, enabling the determination of route types between fixed, hybrid, and flexible routes based on demand, cost, and operational parameters. Analytical results demonstrate that the lower operating costs of SAVs favor more flexible route services. The practical applications and benefits of semi-on-demand feeders are presented with numerical examples and a large-scale case study in the Chicago metropolitan area, USA. Findings reveal scenarios in which flexible route portions serving passengers located further away reduce total costs, particularly user costs, whereas higher demand densities favor more traditional line-based operations. Current cost forecasts suggest smaller vehicles with fully flexible routes are optimal, but operating constraints or higher operating costs would favor larger vehicles with hybrid routes. The study provides an analytical tool to design SAVs as directional services and transit feeders, and tractable continuous approximation formulations for planning and research in transit network design.
\begin{keyword}
semi-on-demand \sep transit feeder \sep shared autonomous vehicle \sep transit design \sep flexible route
\end{keyword}
\end{abstract}

\end{frontmatter}

\makeatletter
\let\orig@makefnmark\@makefnmark

\makeatletter
\let\@makefnmark\relax
\makeatother
\footnotetext{\textsuperscript{\textdagger}Deceased}

\makeatletter
\let\@makefnmark\orig@makefnmark
\makeatother

\setcounter{footnote}{0}

\section{Introduction}
\subsection{Motivation}
Mobility, a cornerstone of socio-economic well-being, is often unevenly distributed due to public transit systems' inconsistent availability, accessibility, and quality. One of the key hurdles limiting transit access outside city centers is the first-mile-last-mile problem --- from a user perspective, public transit is unattractive if reaching it requires walking for an excessive distance from one's origin or to the final destination. From an operator perspective, line-based transit operations are often cost-inefficient in low-density areas with limited ridership. Taxis and rideshare services may plug this gap, but are not economical nor socially desirable from a congestion or environmental standpoint. Furthermore, long-term trends, such as suburbanization, and more recent trends associated with the COVID-19 pandemic, such as reduced service and increased prevalence of remote/hybrid work \citep{tahlyan_for_2022}, have greatly reduced transit ridership and funding. Poor accessibility often deters travelers from using transit, leaving private vehicles as their sole feasible mobility option, despite its higher costs for users (fuel and maintenance), society (delay due to congestion and space requirements for parking), and the environment (15\% of global carbon emissions are contributed by road transport \citep{ritchie_cars_2020}) relative to transit.

Shared Autonomous vehicles (SAVs) are one of the emerging technologies which could alleviate these mobility concerns by offering convenient point-to-point journeys \citep{frei_flexing_2017}. However, this convenience is predicted to induce demand, increasing traffic volume, emissions, and congestion if transit options are not also improved \citep{xu_privately_2019}. This work addresses this challenge by proposing the deployment of SAVs on semi-on-demand hybrid routes, designed to bridge the first-mile-last-mile gap in scenarios delineated herein. Such a system balances the economies of scale of fixed route public transit with the accessibility of flexible route taxis or SAVs, while taking advantage of the lower operating costs of SAVs \citep{tirachini_economics_2020}. Addressing the first-mile-last-mile problem can expand transit coverage to provide seamless door-to-door transportation options to wider populations in a cost-effective, green, and equitable manner. 

\subsection{Problem Description}
We consider a semi-on-demand hybrid route scheduled service along a corridor (Figure~\ref{fig:feeder_illustration}). Each SAV is dispatched on a regular headway, first running on a fixed route
\footnote{A constant stop spacing is assumed for fixed routes in continuous approximation. Optimizing stop spacing for fixed routes is a separate non-trivial problem \citep[e.g.,][]{pandey_bus_2024}, but can be incorporated in the joint optimization framework described later in Section~\ref{sec:op_x_f_s_h_b}.}
for areas close to a train station/downtown (with more concentrated demand), and then, instead of along a fully fixed route, picking up and dropping off passengers on demand in the flexible route portion (with pre-determined length  $x_f$) for areas further from a train station/downtown (with more scattered demand). The demarcation of fixed and flexible service areas by $x_f$ on a particular route is a design variable at a service planning level. While pre-determining $x_f$ for each route reduces real-time dynamic flexibility, it provides a predictable service structure, enabling riders to clearly understand whether to access a fixed stop or await on-demand pickup.

\begin{figure}[htb]
  \centering
  \includegraphics[width=6.5in]{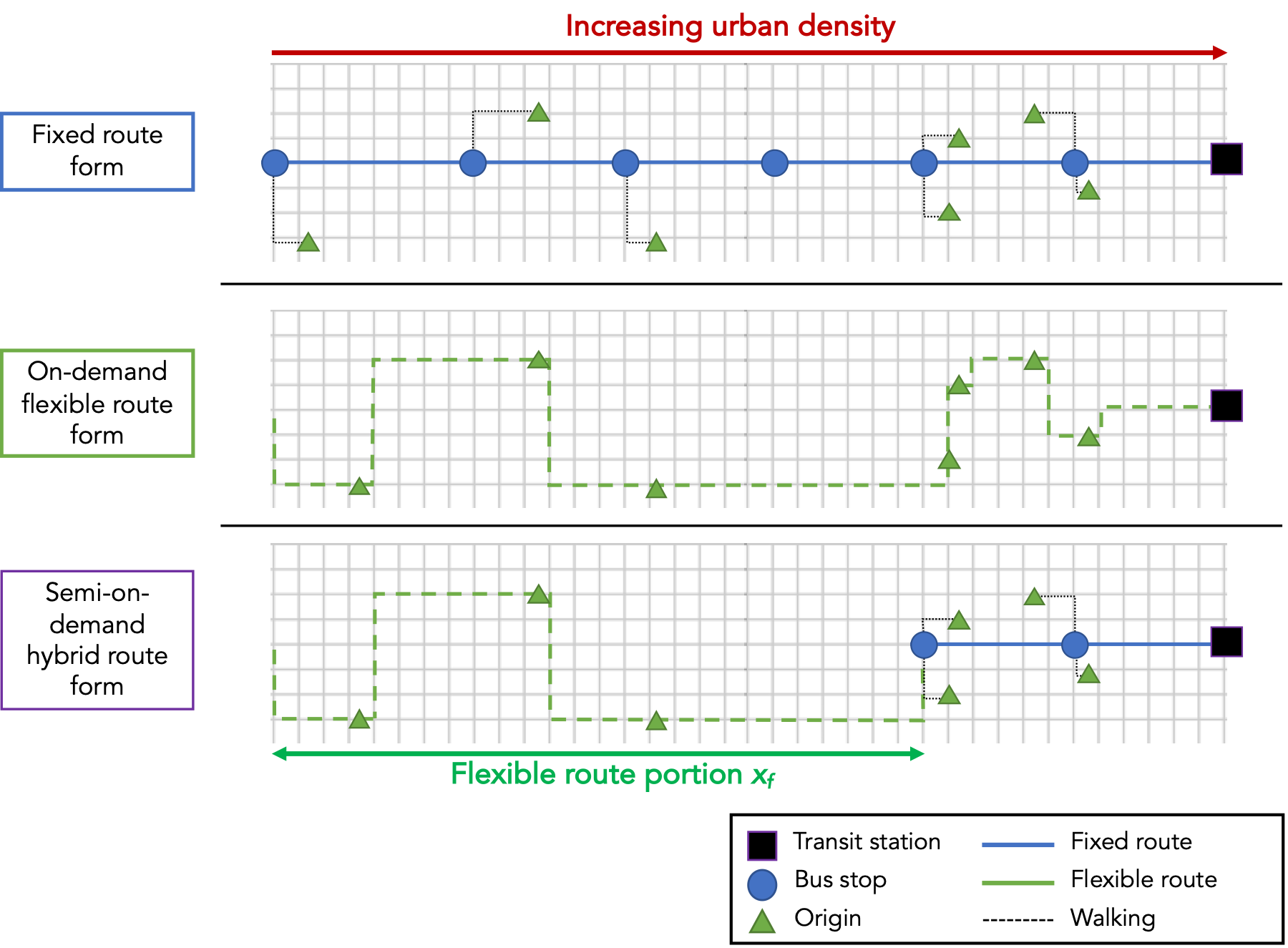}
  \caption{Illustration of Fixed Route, On-demand Flexible Route, and Semi-on-Demand Hybrid (Fixed/Flexible) Route as a Feeder Service}
  \label{fig:feeder_illustration}
\end{figure}

This study focuses on scenarios of directional demand in a relatively narrow corridor of a grid network, such as existing feeder service to train stations or commuter routes from suburbs to downtown. Therefore, SAVs serve all flexible pick-ups/drop-offs sequentially along the traveling direction. While our derivation and results cater to station/downtown-bound services (many-to-one), they equally apply to outbound services (one-to-many). \footnote{Simultaneous many-to-many operations (serving random origins to random destinations within the corridor) would require more complex routing logic and approximation (Volakakis et al. 2023), which would limit the insights drawn from closed-form solutions and is outside the scope of the directional continuous approximation approach in this study.}

We solve for the optimal design variables (flexible route portion $x_f$, fleet size $s$, headway $h$, and vehicle size $b$) analytically and numerically to minimize the total generalized costs that cover users (access, waiting, and riding) and operators (vehicle mileage and fleet requirement) at the route design level. The setting of $x_f$ leads to trade-offs between saving in passengers' access time and additional detours required for both passengers on board and vehicles in the flexible route portion. Assuming known distributions of stochastic and inelastic demand for a route, this study approximates the average detours required in the flexible route portion, thereby optimizing for average system performance (rather than real-time routing for specific instances). The effects of design variables and, more generally, the performance of hybrid route form compared to fixed and flexible routes are investigated with respect to varying parameters of demand, cost, and operations.

\subsection{Objectives and Contributions}

We present two cost formulations for the semi-on-demand hybrid route service: one strategic and one tactical. The strategic formulation optimizes the flexible route portion and fleet size, capturing the additional vehicle requirement for possible detours of a flexible or hybrid route. In addition to these two variables, the tactical formulation optimizes headway for new route planning to balance the effects on the waiting time and vehicle capital cost. It also considers vehicle types (sizes) in the optimization, resulting in a comprehensive model to support network planning, select vehicle and fleet sizes simultaneously, and design flexible route portions and headways.

This paper provides three key contributions. 
\begin{itemize}
    \item Conceptually, we investigate the applicability and benefits of a semi-on-demand hybrid route form with general cost formulations that also cover fixed and flexible route forms. This delineates the conditions under which each of the three route forms is optimal, supplemented with tractable cost approximations. 
    \item Methodologically, we derive an analytical approach to determine the optimal route form and, for hybrid routes, the optimal flexible route portion and fleet size, without the need for computationally expensive simulation-based optimization in initial planning stages. We also conduct joint numerical optimization with the headway and vehicle size. These allow efficient transit service design and sensitivity study with respect to demand, geospatial, service, and cost parameters. 
    \item From an application standpoint, our analytical derivations indicate that the anticipated lower operating costs of SAVs make more flexible services economically advantageous. We demonstrate the benefits and applications of semi-on-demand feeders with numerical examples and a large-scale real-world case study in the Chicago metropolitan area, USA. The analytical formulation provides transit agencies with a tool to assess the benefits of flexible/hybrid/fixed routes for existing routes and jointly consider headway, fleet size, and vehicle size for investment decisions.
\end{itemize}

The paper is structured as follows: Section~\ref{sec:lit_review} reviews relevant literature, focusing on flexible transit service design and SAVs. We then introduce the two mathematical formulations for costs and derive optimal values for the respective decision variables in Section~\ref{sec:math_form}. To illustrate these formulations, we present two numerical examples in Section~\ref{sec:num_eg} --- the first focuses on the flexible route portion and fleet size, and the second incorporates headway and vehicle size. Subsequently, a city-scale case study of feeder services in Chicago is discussed in Section~\ref{sec:case_study}. Finally, we conclude the paper with a summary of the findings, limitations, and future research directions in Section~\ref{sec:concl}.

\section{Background}\label{sec:lit_review}

\subsection{Flexible Transit Route Design}

Previous research efforts mainly focused on flexible/adaptive transit route designs differentiating usage of fixed route and demand-responsive transit (DRT) services. \citet{chang_optimization_1991} analytically compared fixed and one-stop flexible service and optimized fleet and vehicle sizes. \citet{quadrifoglio_methodology_2009} and \citet{ li_feeder_2010} designed analytical models to identify the conditions for switching between demand-responsive and fixed route policies that maximize service quality. \citet{nourbakhsh_structured_2012} analyzed flexible route transit systems in a corridor with continuous approximation. \citet{rich_fixed_2023} compared fixed route and demand-responsive methods as feeders for light rail transit using agent-based simulation. Other studies looked into parallel flexible routes \citep{chen_analysis_2017}, dynamic switching between flexible and fixed route \citep{sayarshad_optimizing_2020}, electric and autonomous fleet management \citep{bongiovanni_electric_2019, bongiovanni_machine_2022}, and joint design of transit and DRT with queuing network model \citep{liu_mobility_2021} and path-based network optimization \citep{steiner_strategic_2020}. More recently, \citet{martinez_mori_value_2023} showed mathematically the value of flexibility brought by DRT. \citet{calabro_adaptive_2023} chose a fixed route or DRT feeder and designed the trunk service accordingly. For other recent advances in demand-responsive systems, see the survey by \citet{vansteenwegen_survey_2022}. Overall, these studies considered decisions to deploy flexible route or fixed route services. In contrast, this study investigates a hybrid route model, conceptualized as a continuous spectrum between fully fixed and fully flexible services, defined by the length of its flexible route portion.

Research on semi-flexible transit systems involved various optimization strategies (see a survey by \citet{errico_survey_2013} on earlier planning efforts). \citet{qiu_dynamic_2014} used an analytical model and simulation to decide which curb-to-curb stops are served, aiming to minimize user costs. \citet{galarza_montenegro_large_2021} proposed a demand-responsive feeder system with mandatory and optional bus stops. \citet{zheng_benefits_2019} introduced a blend of flexible and fixed routes via meeting points. \citet{leffler_simulation_2021} investigated on-demand operational policies for autonomous vehicles, showing improved waiting times on trunk routes in particular. \citet{mishra_effect_2023} optimized service headway and curb-to-curb service requests using non-dominated sorting to generate a Pareto set. For stochastic approaches, \citet{rambha_adaptive_2016} formulated an adaptive routing problem as a finite horizon Markov decision process. \citet{li_frequency-based_2023} looked into a stochastic problem of zonal on-demand flexible bus routes. \citet{silva_-demand_2022} proposed a Markovian continuous approximation-based semi-flexible system model on a single bus operation. Other similar service models include "customized bus" with demand-based pre-designed routes \citep{huang_two-phase_2020,abdelwahed_balancing_2023}, hybrid first-mile-last-mile service between shuttle and transportation network companies (TNCs) \citep{grahn_optimizing_2022}, and demand-adaptive system \citep{errico_single-line_2021}. While the previously discussed literature explored a wide range of flexible transit services, this study focuses on the proposed semi-on-demand routes, which pre-design portions of flexible and fixed routes for consistent service patterns. It combines the economies of scale and predictability of scheduled fixed routes and the door-to-door convenience of flexible routes. Additionally, most previous efforts relied on a mix of simulation, heuristics, or mixed-integer linear programming, often limiting intuitive sensitivity analyses with respect to key design parameters. This study provides an analytical approach to total cost minimization for delineating the conditions of the optimal route form (fixed/hybrid/flexible) with the optimal flexible route portion and fleet size.

From the perspective of ride-pooling, research have looked into incorporating meeting points (or virtual stops to minimize detours and improve operational efficiencies. Studies have shown that introducing meeting points can significantly increase matching rates and mileage savings \citep{stiglic_benefits_2015}, and optimizing these flexible boarding locations can reduce fleet distance traveled while maintaining service levels \citep{engelhardt_benefits_2021}. While this study focuses on the distinction between fixed and flexible routes and assumes the meeting points as the stops on the fixed routes (high-density corridor), the integration of such aggregation strategies, similar to the ``bus tubes'' with virtual stops discussed by \citet{shen_design_2021} or the high-coverage hubs proposed by \citet{jung_high-coverage_2011}, would enable more flexible route portions by reducing the marginal cost of detours.

\subsection{Shared Autonomous Mobility Services in Public Transit Systems}

The autonomous driving capability of SAVs leads to lower operating costs \citep{becker_impact_2020, horl_fleet_2019, narayanan_shared_2020}, while the shared element of SAVs improves service quality and operational efficiency \citep{hyland_operational_2020}. Their potential synergies as part of the public transit systems are highlighted by \citet{shen_integrating_2018}, \citet{salazar_intermodal_2020}, and \citet{cortina_fostering_2023} in case studies of Singapore, New York, and Lyon. \citet{badia_design_2021} designed feeder systems with SAVs and studied costs with continuous approximation. \citet{gurumurthy_first-mile-last-mile_2020} demonstrated the potential of SAVs as feeders for public transit in Austin, proposing fare benefits for transit users to increase transit coverage and reduce walking distances. \citet{levin_linear_2019} looked into the optimal integration of SAVs and transit with continuous approximation and linear programming. Transit was only used when it reduced total system travel time. \cite{luo_multimodal_2021} optimized mobility-on-demand service flow, transit frequency, and pricing jointly. \cite{ng_redesigning_2024, ng_joint_2025} redesigned existing multimodal transit networks with SAVs as point-to-point and feeder services. Some research (e.g., \citet{zhang_modular_2020, liu_improving_2021, tian_planning_2022}) investigated an evolving idea of autonomous modular transit and their roles in the overall network.

A contrasting body of research evaluated SAVs as replacements for existing transit systems. \citet{sieber_improved_2020} proposed that autonomous mobility on demand (AMoD) could potentially replace rural trains in Switzerland, analyzing the difference in costs and service levels. \citet{ng_autonomous_2023} and \citet{ volakakis_city-wide_2023} showed around 20\% performance improvement by replacing fixed bus routes with flexible routes in suburbs. \citet{mo_competition_2021} analyzed the impact of dynamic adjustable supply strategies and regulations in SAV fleet size and transit headway changes on system efficiency. \citet{cao_autonomous_2019} proposed stop-skipping with autonomous shuttles and evaluated the number of extra vehicles required.

Lastly, several researchers have sought to optimize SAV design variables, such as vehicle size \citep{alonso-mora_-demand_2017} and fleet size \citep{pinto_joint_2020, dandl_regulating_2021}. \citet{pinto_joint_2020} highlighted the need to consider waiting time in optimization, presenting a bi-level mathematical programming formulation and solution approach. \citet{sadrani_optimization_2022} optimized autonomous bus service frequency and vehicle size. Similar to a previous work \citep{ng_autonomous_2023} which assessed the attractiveness of flexible route operation with autonomous minibuses in suburbs, this study proposed a generalized framework to optimize the portion of flexible route with fleet size, headway, and vehicle size. Specifically, our methodology determines the optimal division between fixed and flexible route segments along this continuous spectrum of hybrid routes.

\section{Mathematical Formulations}\label{sec:math_form}
This section consists of two main models. We start with the optimization of the flexible route portion $0 \leq x_f \leq L_x$ ($L_x$ is the total route length in the x-direction in km.) and fleet size $s>0$ (in vehicle, or veh) subject to a fixed headway $H$ (in h), which allows us to set $s$ as a function of $x_f$. We discuss the model first with a general demand distribution, followed by specific forms under uniform and triangular distributions. It is worth noting that $x_f=0$ implies a fixed route, $0<x_f<L_x$ a hybrid route, and $x_f=L_x$ a flexible route. Next, we relax the fixed headway assumption and set it as a function $h(x_f,s)$ in optimization. This accounts for the trade-off between waiting times and extra vehicle costs. We also generalize vehicle sizes where the previously constant operational and vehicle costs are functions of vehicle sizes under a capacity requirement constraint.

The general assumptions throughout this section are as follows:
\begin{enumerate}

\item We illustrate the service in a grid network with route running direction as the x-axis and perpendicular detours as the y-axis. Vehicles travel along the x-direction and stop by demand points in sequence, i.e., there is no backtracking.
\item The total dwell time and layover time (at the transit station) are constant; vehicles travel at a constant speed $V_d$ (in km/h).
\item Demand is considered static within a continuous planning horizon and is assumed to be serviceable by the planned capacity.\footnote{This is further addressed by a capacity buffer $\rho$ in Section~\ref{sec:op_x_f_s_h_b}.}
\item All travelers are individuals; group size is not considered other than as a collection of independent travelers.
\end{enumerate}

The notation is summarized in Table~\ref{tab:notation}. Variables are denoted as small Roman characters (e.g., $x_f$, $s$, $b$, $h$) and constants as capital Roman characters (e.g., $L_x$, $H$), except for cost coefficients $\gamma$ and total demand $\Lambda$. The demand density function is $f(x) \geq 0, \forall x\in[0,L_x]$ (in passenger/km or pax/km), with its derivative as $f'(x)$ and cumulative density function as $F(x) = \int_0^x f(\tilde{x}) d\tilde{x}$ (in pax). The total demand is denoted as $\Lambda = F(L_x)$ (in pax). Demand and costs are on a per-hour basis, for example, $\Lambda$ in passengers per hour (pax/h) and $c$ in dollars per hour (\$/h).

We conventionally go from flexible route to fixed route along the x-axis (see Figure~\ref{fig:feeder_illustration}; this generally suggests going from lower- to higher-density areas but imposes no restrictions on demand distribution). Demand along the y-axis is also assumed to follow a general distribution, resulting in a mean fixed route access time $\overline{t_a}$ (in h) and a mean detour $\overline{d_{d,y}}$ (in km), which is the average y-directional distance between consecutive request points along the x-direction.\footnote{For continuous approximation of $\overline{d_{d,y}}$ using mean absolute difference under uniform and normal distributions, see \citet{ng_autonomous_2023}).}

\begin{table}[htb]
\centering
\caption{Summary of Notation}
\label{tab:notation}
\begin{tabular}{|l|l|l|}
\hline
\textbf{Symbol}      & \textbf{Description}                                 & \textbf{Unit} \\ \hline
$\gamma_a$           & Access cost factor                                   & -             \\ \hline
$\gamma_o$           & Operating cost                                       & \$/km         \\ \hline
$\gamma_t$           & Value of time                                        & \$/h          \\ \hline
$\gamma_v$           & Vehicle cost                                         & \$/veh-h      \\ \hline
$\gamma_w$           & Waiting cost factor                                  & -             \\ \hline
$\Lambda$            & Total demand                                         & pax/h         \\ \hline
$\lambda$            & Demand density                                       & pax/km-h      \\ \hline
$\lambda_0$          & Maximum demand density (for triangular distribution) & pax/km-h      \\ \hline
$\rho$               & Capacity buffer over demand                          & -             \\ \hline
$\mathcal{B}$        & Set of vehicle sizes                                         & -           \\ \hline
$b$                  & Vehicle size                                         & pax           \\ \hline
$c$                  & Total cost                                           & \$/h          \\ \hline
$c_{o,x}$            & Operator cost (x-directional)                        & \$/h          \\ \hline
$c_{o,y}$            & Operator cost (y-directional)                        & \$/h          \\ \hline
$c_{t,x}$            & Riding cost (x-directional)                          & \$/h          \\ \hline
$c_{t,y}$            & Riding cost (y-directional)                          & \$/h          \\ \hline
$c_a$                & Access cost                                          & \$/h          \\ \hline
$c_v$                & Vehicle cost                                         & \$/h          \\ \hline
$c_w$                & Waiting cost                                         & \$/h          \\ \hline
$\overline{d_{d,y}}$ & Mean detour                                          & km            \\ \hline
$f(\cdot)$                  & Demand density function                              & pax/km-h      \\ \hline
$F(\cdot)$                  & Cumulative demand density function                   & pax/h         \\ \hline
$H$ or $h$           & Headway                                              & h             \\ \hline
$L_x$                & Total route length in the x-direction                & km            \\ \hline
$s$                  & Fleet size                                           & veh           \\ \hline
$\overline{t_a}$     & Mean fixed route access time                         & h             \\ \hline
$V_d$                & Vehicle speed                                        & km/h          \\ \hline
$x_f$                & Flexible route portion                               & km            \\ \hline
\end{tabular}
\end{table}

\subsection{Optimization of the Flexible Route Portion and Fleet Size given a Headway}\label{sec:op_x_f_s}
\subsubsection{General Form}

In this formulation, we assume a general demand distribution along the x-direction and a constant headway $H$ (in h) to maintain constant waiting time and route capacity, and ignore the fleet size integrality\footnote{Non-integer fleet size may be used to represent operational arrangement, e.g., interlining, where vehicles are shared across routes. Furthermore, the additional complexity induced by the integrality constraints may necessitate numerical optimization methods.}. We first approximate the y-directional total detour time $t_y(x_f)$ in Eq.~\ref{eq:t_y}, based on the number of detours, i.e., demand in the flexible route portion $F(x_f)$, multiplied by $H$ for each vehicle trip and the average detour time $\overline{d_{d,y}}/v_d$. Then, to simplify the minimization problem to univariate, the fleet size required $s$ can be expressed as a function of $x_f$ by considering the cycle time divided by the headway $H$ in Eq.~\ref{eq:s_fleet}, where $t_x$ is the x-directional travel time and $T_l$ is the layover time between one-way bus trips.

\begin{align}
\label{eq:t_y}
t_y(x_f) &= H \frac{\overline{d_{d,y}}}{V_d} F(x_f)
\\
\label{eq:s_fleet}
s(x_f) &= \frac{2}{H} \left( t_x + t_y (x_f) + T_l \right) 
\nonumber 
\\ 
&= \frac{2}{H} \left( \frac{L_x}{V_d} + H \frac{ \overline{d_{d,y}}}{V_d} F(x_f) + T_l \right)
\end{align}

We consider total cost from a societal perspective encompassing impacts on users and the operator with reference to formulations of \citet{newell_issues_1979} and \citet{ng_autonomous_2023}. The objective is to minimize the total costs $c(x_f)$ (in \$) as a function of the flexible route portion $x_f$. It is composed of travelers' costs: access cost $c_a(x_f)$ (for fixed route portion only), waiting cost $c_w$, riding costs (x-direction: $c_{t,x}$, and, for flexible route portion only, y-directional detour $c_{t,y}(x_f)$); and distance-based operator costs:  $c_{o,x}$ (in x-direction) and $c_{o,y}(x_f)$ (y-direction, for flexible route portion only) and vehicle cost $c_v(s)$ as a function of the fleet size $s>0$. The cost factors are the value of time $\gamma_t$ (in \$/h), access cost factor $\gamma_a$ and waiting cost factor $\gamma_w$ (in multiples of riding cost), as well as operating cost $\gamma_o$ (in \$/km) and vehicle cost $\gamma_v(s)$ (in \$/veh).

The access (walking) cost $c_a(x_f)$ is calculated in Eq.~\ref{eq:c_ac} as the product of the value of access time, $\gamma_t\gamma_a$, average access time $\overline{t_a}$, and number of passengers served in the fixed route portion $\Lambda - F(x_f)$. The waiting cost $c_w$ in Eq.~\ref{eq:c_wa} is a constant equal to the product of values of waiting time $\gamma_t\gamma_w$, total number of passengers $\Lambda$, and expected waiting time (half of headway $H$).

\begin{align}
c_a(x_f) &= \gamma_t \gamma_a \overline{t_a} (\Lambda - F(x_f))
\label{eq:c_ac}
\\
c_w &= \gamma_t \gamma_w \Lambda \frac{H}{2}
\label{eq:c_wa}
\end{align}

The riding cost is separated into x-directional travel and y-directional detours. For x-directional travel, each passenger getting on the bus at $x$ needs to travel a distance of $L_x-x$, resulting in the total riding distances as the integral $\int_0^{L_x} f(x) (L_x-x) dx$ for $f(x)$ passengers in an hour. The constant x-directional riding cost $c_{t,x}$ (hourly) is then its product with the value of time $\gamma_t$ divided by the vehicle speed $V_d$ in Eq.~\ref{eq:c_t_x}. For y-directional detour, each passenger getting on the bus at $x$ in the flexible portion rides the extra detours to pick up the remaining passengers between $x$ and $x_f$ in the flexible portion, i.e., $\int_{x}^{x_f} H f(\tilde{x}) d\tilde{x}$ (multiplied by $H$ (hours per vehicle) to convert hourly demand (passengers per hour) to the number of passengers per vehicle). This extra detour experienced by each passenger boarding at each $x \in [0, x_f]$ in the flexible route region is multiplied by the demand density $f(x)$, resulting in the double integral in Eq.~\ref{eq:c_t_y}. After simplification, it is apparent that the y-directional riding cost $c_{t,y}(x_f)$ is proportional to the square of the cumulative demand function $F(x_f)$.

\begin{align}
c_{t,x} &= \gamma_t \frac{1}{V_d} \int_0^{L_x} f(x) (L_x-x) dx 
\nonumber 
\\ 
&= \gamma_t \frac{1}{V_d} \int_0^{L_x} F(x) dx 
\label{eq:c_t_x}
\\
c_{t,y}(x_f) &= \gamma_t \int_{0}^{x_f} f(x) \int_{x}^{x_f} H f(\tilde{x}) d\tilde{x} dx \frac{\overline{d_{d,y}}}{V_d} 
\nonumber 
\\ 
&= \gamma_t H \frac{\overline{d_{d,y}}}{2V_d} F(x_f)^2 
\label{eq:c_t_y}
\end{align}

The distance-based operating cost is also separated into x- and y-directional components. The constant x-directional operating cost $c_{o,x}$ in Eq.~\ref{eq:c_o_x} is determined by the route length, divided by the headway $H$ to reflect the hourly frequency. In contrast, for the y-directional detour, $H$ is canceled out for $c_{o,y}(x_f)$ in Eq.~\ref{eq:c_o_y}, i.e., the hourly operating cost due to detours depends on the total number of detours in an hour but not the headway.

\begin{align}
c_{o,x} &= \gamma_o \frac{L_x}{H}
\label{eq:c_o_x}
\\
c_{o,y}(x_f) &= \gamma_o \overline{d_{d,y}} F(x_f)  
\label{eq:c_o_y}
\end{align}

The vehicle cost is calculated in Eq.~\ref{eq:c_v} through Eq.~\ref{eq:s_fleet} as directly proportional to the hourly vehicle cost  $\gamma_v$ (in \$/veh), to cover the marginal costs of additional vehicles.

\begin{align}
\label{eq:c_v}
c_v(x_f) &= \gamma_v s(x_f) = \gamma_v \frac{2}{H} \left( \frac{L_x}{V_d} + \frac{H \overline{d_{d,y}}}{V_d} F(x_f) + T_l \right)
\end{align}

The total cost $c(x_f)$ in Eq.~\ref{eq:c_gen} is the sum of all costs in Eq.~\ref{eq:c_ac}-\ref{eq:c_v}. It enables us to determine the optimal route form analytically with Result~\ref{thm:route_choice} based on geospatial, cost, and operational parameters.

\begin{align}
c(x_f) &= \gamma_t \gamma_a \overline{t_a} (\Lambda - F(x_f)) 
+ \gamma_t \gamma_w \Lambda \frac{H}{2} 
\nonumber \\
&\quad + \gamma_t \frac{1}{V_d} \int_0^{L_x} F(x_f) dx 
+ \gamma_t H \frac{\overline{d_{d,y}}}{2V_d} F(x_f)^2 
\nonumber \\
&\quad + \gamma_o \frac{L_x}{H} 
+ \gamma_o \overline{d_{d,y}} F(x_f) 
\nonumber \\
&\quad + \gamma_v \frac{2}{H} \left( \frac{L_x}{V_d} + \frac{H \overline{d_{d,y}}}{V_d} F(x_f) + T_l \right)
\label{eq:c_gen}
\end{align}

\begin{thm}
\label{thm:route_choice}
Given a fixed fleet size and a fixed headway $H$, the optimal route form that minimizes the total cost function $c(x_f)$ in Eq.~\ref{eq:c_gen} is one of the following:
\begin{enumerate}[label=(\alph*)]
    \item Fixed route if  $ \overline{t_a} / \overline{d_{d,y}} \leq \gamma_o / (\gamma_t \gamma_a) + 2 \gamma_v / (\gamma_t \gamma_a V_d)$;
    \item Flexible route if  $ \overline{t_a} / \overline{d_{d,y}} \geq H \Lambda / (\gamma_a V_d) + \gamma_o / (\gamma_t \gamma_a) + 2 \gamma_v / (\gamma_t \gamma_a V_d)$; or
    \item Hybrid route if $ \gamma_o / (\gamma_t \gamma_a) + 2 \gamma_v / (\gamma_t \gamma_a V_d)< \overline{t_a} / \overline{d_{d,y}} < H \Lambda / (\gamma_a V_d) + \gamma_o / (\gamma_t \gamma_a) + 2 \gamma_v / (\gamma_t \gamma_a V_d)$, in which case Eq.~\ref{eq:F_gen} determines the optimal cumulative demand within the flexible route portion $F(x_f^*)$ and subsequently the optimal flexible route length $x_f^*$.
\end{enumerate}

\begin{align}
F(x_f^*) &= \frac{1}{H} \left( \gamma_a \frac{V_d \overline{t_a}}{\overline{d_{d,y}}} - \frac{\gamma_o}{\gamma_t} V_d - 2\frac{\gamma_v}{\gamma_t}\right)
\label{eq:F_gen}
\end{align}
\end{thm}

\begin{proof}

\begin{enumerate}[label=(\alph*)]
    \item Fixed route: Assume that the minimum point $x_f^*>0$, so $\exists x_f\in(0,L_x]: c(x_f)<c(0)$. The condition can be rearranged as  $ \gamma_t \gamma_a \overline{t_a} \leq \gamma_o \overline{d_{d,y}} + 2 \gamma_v \overline{d_{d,y}} / V_d$ , so the first derivative in Eq.~\ref{eq:c'_gen} $c'(x_f)>0, \forall x_f \in (0,L_x]$ as $F(x_f)>0$. Therefore, $c(x_f) > c(0), \forall x_f \in (0,L_x]$. Contradiction shows that $x_f^* = 0$.
    \item Flexible route: Similar to Part (a), assume that the minimum point $x_f^* < L_x$, so $\exists x_f\in[0,L_x): c(x_f)<c(L_x)$. The condition can be rearranged as $ \gamma_t \gamma_a \overline{t_a} \geq \gamma_t H \Lambda \overline{d_{d,y}} / V_d + \gamma_o \overline{d_{d,y}} + 2 \gamma_v \overline{d_{d,y}} / V_d$, so $\gamma_t \gamma_a \overline{t_a} > \gamma_t H \overline{d_{d,y}} F(x_f) / V_d + \gamma_o \overline{d_{d,y}} + 2 \gamma_v \overline{d_{d,y}} / V_d, \forall x_f \in [0,L_x) $ as $\Lambda > F(x_f)$. Therefore, $c'(x_f) < 0, \forall x_f \in [0,L_x)$ in Eq.~\ref{eq:c'_gen}, suggesting $c(x_f) > c(L_x), \forall x_f \in [0,L_x)$. Contradiction shows that $x_f^* = L_x$.
    \item     Hybrid route: We minimize the total cost with respect to $x_f$ by considering the optimality condition where $c'(x)=0$ in Eq.~\ref{eq:c'_gen}, resulting in $F(x_f^*)$ in Eq.~\ref{eq:F_gen}. We note that $F(x_f) < F(x_f^*)$ implies $c'(x_f) < 0$ and $F(x_f) > F(x_f^*)$ implies $c'(x_f) > 0$. Additionally, the second derivative at this point is positive from Eq.~\ref{eq:c''_gen_op}, so the cost function is convex around the optimum point. Hence, $c(x_f^*)$ is minimal when $x_f^*$ satisfies Eq.~\ref{eq:F_gen}. Lastly, the case condition implies $0 < F(x_f^*) < \Lambda$, so $0<x_f^*<L_x$. 

\end{enumerate}

\begin{align}
c'(x_f) &= 
- \gamma_t \gamma_a \overline{t_a} f(x_f) 
+ \gamma_t  H \frac{\overline{d_{d,y}}}{V_d} f(x_f) F(x_f) 
\nonumber \\
&\quad + \gamma_o \overline{d_{d,y}} f(x_f) 
+ 2 \gamma_v \frac{\overline{d_{d,y}}}{V_d} f(x_f) 
\label{eq:c'_gen}
\\
c''(x_f) &= 
- \gamma_t \gamma_a \overline{t_a} f'(x_f) 
\nonumber \\
&\quad+ \gamma_t  H \frac{\overline{d_{d,y}}}{V_d} (f'(x_f) F(x_f) + (f(x_f))^2) 
\nonumber \\
&\quad + \gamma_o \overline{d_{d,y}} f'(x_f) 
+ 2 \gamma_v \frac{\overline{d_{d,y}}}{V_d} f'(x_f)
\label{eq:c''_gen}
\\
c''(x_f^*) &= \gamma_t H \frac{\overline{d_{d,y}}}{V_d} (f(x_f^*))^2 \geq 0
\label{eq:c''_gen_op}
\end{align}

\end{proof}

Eq.~\ref{eq:F_gen} offers an intuitive interpretation of the hybrid route design. It suggests that the first $F(x_f^*)$ passengers (starting from the low-density end of the route) should be served via flexible routing, regardless of the specific x-directional demand distribution $f(\cdot)$.\footnote{While the quantity of flexible demand $F(x_f^*)$ is determined by cost and operational parameters in Eq.~\ref{eq:F_gen}, the spatial extent $x_f^*$ depends on the demand density. For a denser demand distribution (larger $f(x)$), the cumulative demand threshold $F(x_f^*)$ is reached over a shorter distance, resulting in a smaller $x_f^*$. Conversely, lower demand density implies a longer flexible route portion, which aligns with the condition in Result~\ref{thm:route_choice}(b) where a smaller total demand $\Lambda$ favors fully flexible routes.} This threshold is reached when the marginal savings in access cost (walking time) for a passenger no longer outweigh the marginal costs incurred by the system (additional riding time for passengers on board and operational costs for the vehicle). The number increases with the vehicle speed $V_d$, access cost $\gamma_a \overline{t_a}$ (walking time), and value of time $\gamma_t$, and decreases with the headway $H$, mean detour $\overline{d_{d,y}}$, operating cost $\gamma_o$, and vehicle cost $\gamma_v$. This aligns with the trade-offs of savings in access cost with extra riding time and operational costs incurred by detours.

Furthermore, the expression explicitly highlights the role of vehicle technology. The term $\gamma_o / \gamma_t$ represents the ratio of operating cost to the value of time. High operating costs (typical of human-driven vehicles) increase the negative terms in Eq.~\ref{eq:F_gen}, potentially driving $F(x_f^*)$ to zero (favoring fixed route). Conversely, SAVs, which promise significantly lower $\gamma_o$, reduce this penalty, mathematically justifying a shift toward larger flexible route portions ($x_f^*$) or fully flexible services.

Simplified demarcation criteria for optimal route form in different design parameters are shown in Figure~\ref{fig:route_choice_chart}, as derived from Result~\ref{thm:route_choice} by assuming negligible effect of the flexible route on the fleet size $s$ and omitting vehicle cost ($\gamma_v=0$). The rearranged condition for fixed route $ \gamma_t \gamma_a \overline{t_a} \leq \gamma_o \overline{d_{d,y}}$ suggests that for each detour, the access cost saving is not greater than the extra operational cost brought. For flexible route, $ \gamma_t \gamma_a \overline{t_a} \geq \gamma_t H \Lambda \overline{d_{d,y}} / V_d + \gamma_o \overline{d_{d,y}}$ means that the access cost saving is not smaller than the sum of extra operational and riding cost (for every other passenger) brought by each detour. This also delineates the hybrid route case where the access cost saving is greater than the operational cost but not enough to cover the extra riding cost for all passengers. Compared to the previous results by \citet{ng_autonomous_2023} (shown by the gray dashes), Result~\ref{thm:route_choice} demarcates the previous marginal cases between fixed and flexible routes where the hybrid route would provide better services. Eq.~\ref{eq:F_gen} effectively means the first $F(x_f^*)$ passengers should be served with a flexible route, regardless of the specific x-directional demand distribution. 

\begin{figure}[htb]
  \centering
  \includegraphics[width=4in]{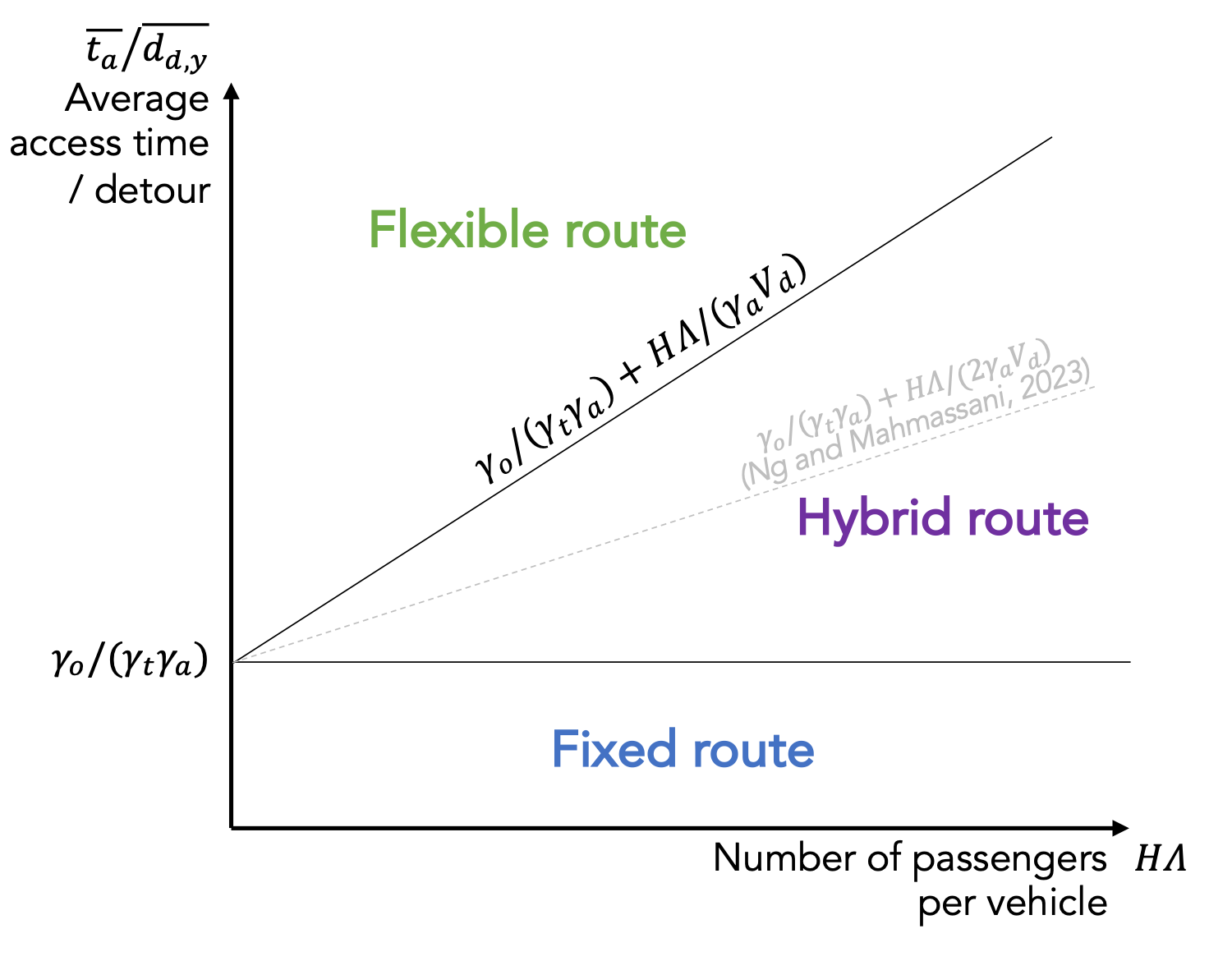}
  \caption{Illustration of Optimality Demarcation of Fixed, Hybrid, and Flexible Route from Result~\ref{thm:route_choice} (Ignoring Additional Vehicle Requirement)}
  \label{fig:route_choice_chart}
\end{figure}

On the effects of vehicle technology, visually in Figure~\ref{fig:route_choice_chart}, lower operating costs would shift the y-intercepts of fixed/hybrid and hybrid/flexible route demarcation down. This again suggests SAVs, with their lower operating cost compared to human drivers, favor offering more flexible or hybrid route services rather than traditional fixed routes.

Lastly, the vehicle cost $\gamma_v$ suggests that extra vehicles needed for the detours in flexible routes would elevate the cost and limit the extent of the flexible route at optimum. Determining the optimal fleet size conditioned on the optimal flexible portion leads to Result~\ref{cor:fleet_op}.

\begin{thm}
\label{cor:fleet_op}
For a hybrid route (with a fixed headway $H$), the optimal fleet size $s^*$ can be expressed in Eq.~\ref{eq:s_fleet_op}, which is independent of the x-directional demand distribution $f(x)$.
\begin{align}
\label{eq:s_fleet_op}
s^* &= \frac{2}{H} \left( \frac{L_x}{V_d} +  \gamma_a \overline{t_a} - \frac{\gamma_o}{\gamma_t} \overline{d_{d,y}} - 2\frac{\gamma_v}{\gamma_t} \frac{\overline{d_{d,y}}}{V_d} + T_l \right)
\end{align}
\end{thm}
\begin{proof}
Eq.~\ref{eq:s_fleet_op} can be obtained by combining Eq.~\ref{eq:s_fleet} and \ref{eq:F_gen}. 
\end{proof}

The second term of Eq.~\ref{eq:s_fleet_op} reflects the balance between riders' access time and operator's cost. $s^*$ increases with access time $\overline{t_a}$ as it favors longer flexible routes that would require more vehicles, and decreases with operating cost $\gamma_o$ and y-directional detour $\overline{d_{d,y}}$, which makes it more costly to run flexible service. This aligns with the factors in Eq.~\ref{eq:F_gen} in Result~\ref{thm:route_choice}.

\subsubsection{Uniform Demand Distribution}
We now show specific examples of x-directional demand distribution. Assuming that the demand density follows a uniform distribution with density $\lambda$, i.e., $f(x) = \lambda$, we get $F(x) = \lambda x$ and $\Lambda = \lambda L_x$. $c(x_f)$ in Eq.~\ref{eq:c_gen} is then simplified in Eq.~\ref{eq:c_uni}.

\begin{align}
c(x_f) &= \gamma_t \gamma_a \overline{t_a} \lambda (L_x - x_f)  + \gamma_t \gamma_w \lambda L_x \frac{H}{2} 
\nonumber \\
&\quad + \gamma_t \frac{\lambda L_x^2}{2V_d}
+ \gamma_t \frac{\lambda^2}{2} H \frac{\overline{d_{d,y}}}{V_d} x_f^2
\nonumber \\
&\quad + \gamma_o \frac{L_x}{H} 
+ \gamma_o \lambda \overline{d_{d,y}} x_f  
\nonumber \\
&\quad + \gamma_v \frac{2}{H} \left(\frac{L_x}{V_d} + \frac{H \overline{d_{d,y}}}{V_d} \lambda x_f + T_l \right)
\label{eq:c_uni}
\end{align}

The optimal flexible route portion under a uniform distribution $x_{f,uni}^*$ in Eq.~\ref{eq:x_f_uni} is then obtained by solving Eq.~\ref{eq:F_gen} with $F(x_{f,uni}^*)=\lambda x_{f,uni}^*$.

\begin{align}
x_{f,uni}^* &= \frac{1}{\lambda H} \left( \gamma_a\frac{ V_d \overline{t_a}}{\overline{d_{d,y}}} - \frac{\gamma_o}{\gamma_t} V_d - 2\frac{\gamma_v}{\gamma_t} \right)
\label{eq:x_f_uni}
\end{align}

\subsubsection{Triangular Demand Distribution}
We may also assume a more likely scenario --- increasing demand density closer to the transit station or downtown (see Figure~\ref{fig:feeder_illustration}). For a triangular distribution with demand density $f(x)$ starting with $0$ and increasing linearly with $x$ to $\lambda_0$ at the end, $f(x) = \lambda_0 / L_x \cdot x$, $F(x) = \lambda_0 / 2L_x \cdot x^2$, and $\Lambda = \lambda_0L_x / 2$. The total cost function $c(x_f)$ and the optimal flexible route portion under a triangular distribution $x_{f,tri}^*$  can be obtained similarly in Eq. \ref{eq:c_tri} and \ref{eq:x_f_tri}.

\begin{align}
c(x_f) &= \gamma_t \gamma_a \overline{t_a} \frac{\lambda_0}{2} \left( L_x - \frac{x_f^2}{L_x} \right)  
\nonumber \\
&\quad + \gamma_t \gamma_w \frac{\lambda_0L_x}{2} \frac{H}{2} + \gamma_t \frac{\lambda_0 L_x^2}{6V_d}
+ \gamma_t H\frac{\overline{d_{d,y}}}{V_d} \lambda_0^2 \frac{x_f^4}{8L_x^2} 
\nonumber \\
&\quad + \gamma_o \frac{L_x}{H} 
+ \gamma_o \overline{d_{d,y}} \frac{\lambda_0 x_f^2}{2L_x}
\nonumber \\
&\quad + \gamma_v \frac{2}{H} \left( \frac{L_x}{V_d} + \frac{H \overline{d_{d,y}}}{V_d} \frac{\lambda_0 x_f^2}{2L_x} + T_l \right)
\label{eq:c_tri}
\\
x_{f,tri}^* &= \sqrt{\frac{2L_x}{\lambda_0 H} \left( \gamma_a \frac{V_d \overline{t_a}}{\overline{d_{d,y}}} - \frac{\gamma_o}{\gamma_t} V_d  - 2\frac{\gamma_v}{\gamma_t} \right) }
\label{eq:x_f_tri}
\end{align}
The comparison of the results under the two demand distributions leads to Result~\ref{cor:uni_tri}.

\begin{thm}
\label{cor:uni_tri}
For a hybrid route assuming the same total demand $\Lambda$, the optimal flexible portion under a triangular distribution $x_{f,tri}^*$ is always longer than that under a uniform distribution $x_{f,uni}^*$. Specifically, $x_{f,tri}^* = \sqrt{L_x x_{f,uni}^*}$.
\end{thm}

\begin{proof}
 $x_{f,tri}^* = \sqrt{L_x x_{f,uni}^*}$ is obtained by combining Eq.~\ref{eq:F_gen} and Eq.~\ref{eq:x_f_tri}. Then, $x_{f,tri}^* > x_{f,uni}^*$ for $x_{f,uni}^*<L_x$.
\end{proof}

Result~\ref{cor:uni_tri} highlights that a demand gradient along the route (triangular over uniform), which is common in many transit feeders and city-suburb services, favors longer flexible routes at the end.

\subsection{Joint Optimization of the Flexible Route Portion, Fleet Size, and Headway with Variable Vehicle Sizes}\label{sec:op_x_f_s_h_b}

How would the deployment of CAVs of different vehicle sizes affect the flexible route portion, fleet size, and headway in semi-on-demand routes? We examine this problem by considering the required fleet size and respective operating cost $\gamma_o(b)$ and vehicle cost $\gamma_v(b)$ of each vehicle size $b \in \mathcal{B}$ (in pax/veh), where $\mathcal{B}$ is a set of vehicle sizes. We also relax the assumption of a fixed headway $H$ to allow a variable headway function $h(x_f,s)$. A trade-off is expected between vehicle cost (larger vehicles and less frequent services) and waiting time (longer headway). 

The fleet size $s$ in Eq.~\ref{eq:s_fleet} is rewritten with the headway function $h(x_f,s)$ in Eq.~\ref{eq:s_varH}. After re-arranging the terms, we obtain Eq.~\ref{eq:h_varH}.

\begin{align}
\label{eq:s_varH}
s &= \frac{2}{h(x_f,s)} \left(\frac{L_x}{V_d} + \frac{\overline{d_{d,y}}}{V_d}F(x_f) h(x_f,s) + T_l\right)    
\\
\label{eq:h_varH}
h(x_f,s) &= \frac{L_x + T_l V_d}{s V_d / 2 - \overline{d_{d,y}}F(x_f)}
\end{align}

We also need to update several cost components for the variable headway. Firstly, the waiting cost depends on the headway function in Eq.~\ref{eq:c_w_varH}. Next, riding costs for y-directional detours also vary with the headway in Eq.~\ref{eq:c_t_y_varH}, as fewer passengers in each vehicle mean fewer detours in the flexible portion.

\begin{align}
\label{eq:c_w_varH}
c_w(x_f,s) &= \gamma_t \gamma_w \Lambda \frac{h(x_f,s)}{2}
\\
\label{eq:c_t_y_varH}
c_{t,y}(x_f,s) &= \gamma_t \frac{\overline{d_{d,y}}}{2V_d} F(x_f)^2 h(x_f,s)
\end{align}

For operating costs, the x-component of operating costs depends on $h(x_f,s)$ in Eq.~\ref{eq:c_o_x_varH} because of the service frequency change. However, the y-component is unaffected due to the constant number of detours per hour. Lastly, the vehicle cost can no longer be simply expressed as a function of $x_f$ as in Eq.~\ref{eq:c_v}, but is still directly proportional to the fleet size in Eq.~\ref{eq:c_v_varH}. 

\begin{align}
\label{eq:c_o_x_varH}
c_{o,x}(x_f,s) &= \gamma_o(b) \frac{L_x}{h(x_f,s)}
\\
\label{eq:c_v_varH}
c_v(s) &= \gamma_v(b) s
\end{align}

As vehicle sizes and corresponding costs are discrete, it provides an opportunity to enumerate $b$ and solve for the optimal $x_f$ and $s$ that minimize total costs. To account for demand variability and ensure adequate service quality, we impose a capacity requirement constraint that the provided hourly capacity $b/h(x_f,s)$ has to be greater than the demand $\Lambda$ by a buffer of $\rho$, i.e., $\rho b / h(x_f,s) \geq \Lambda$. Combined with Eq.~\ref{eq:h_varH}, this results in a lower bound of the fleet size $s$ for each vehicle capacity $b$ in Eq.~\ref{eq:s_LB}.

\begin{align}
\label{eq:s_LB}
s \geq \frac{2}{V_d} \left( \frac{\Lambda}{\rho b} (L_x + T_l V_d) + \overline{d_{d,y}}F(x_f) \right)
\end{align}

We can formulate a mathematical problem (Eq.~\ref{eq:c_veh_size}) to minimize the total cost $c(x_f,s,b)$ with respect to $x_f$, $s$, and $b$, subject to the capacity requirement constraint in Eq.~\ref{eq:s_LB}. The analytical expression of optimal $x_f$ and $s$ is complicated, so it is more practical to solve the minimization numerically. For each discrete $b$, we can solve for $(x_f^*,s^*)$ that minimize $c(x_f,s,b)$ and subsequently apply Eq.~\ref{eq:h_varH} to obtain the optimal headway $h(x_f^*,s^*)$. We note that Result~\ref{thm:route_choice} still applies here for a determined headway, which is useful for conducting analysis but not obtaining a solution. 

\begin{align}
\label{eq:c_veh_size}
&\min_{x_f,s,b} c(x_f,s,b) 
\nonumber \\
&=
\gamma_t \gamma_a \overline{t_a} (\Lambda - F(x_f)) + \gamma_t \gamma_w \Lambda \frac{h(x_f,s)}{2} 
\nonumber \\
&\quad + \gamma_t \frac{1}{V_d} \int_0^{L_x} F(x) dx
+ \gamma_t h(x_f,s) \frac{\overline{d_{d,y}}}{2V_d} F(x_f)^2 
\nonumber \\
&\quad + \gamma_o(b) \frac{L_x}{h(x_f,s)} 
+ \gamma_o(b) \overline{d_{d,y}} F(x_f)
\nonumber \\
&\quad + \gamma_v(b) \frac{2}{h(x_f,s)} \left( \frac{L_x}{V_d} + \frac{h(x_f,s) \overline{d_{d,y}}}{V_d} F(x_f) + T_l \right)
\nonumber \\
&\text{subject to} ~~~~~ \text{Eq.}~(\ref{eq:s_LB}), ~~~ 0 \leq x_f \leq L_x, ~~~ b \in \mathcal{B}
\end{align}

\section{Numerical Examples}\label{sec:num_eg}

To demonstrate the mathematical models and benefits of the semi-on-demand routes, we present two sets of numerical examples, each with two Chicago bus routes illustrated in Figure~\ref{fig:bus_map} (also studied by \citet{ng_autonomous_2023}). The first example is the joint optimization of the flexible route portion $x_f$ and fleet size $s$ with a given headway $H$ (Section~\ref{sec:op_x_f_s}), and the second is the joint optimization of flexible route portion $x_f$, fleet size $s$, headway $h$, and vehicle size $b$ (Section~\ref{sec:op_x_f_s_h_b}). The two Chicago Transit Authority (CTA) bus routes studied, CTA126 and CTA84, mainly connect passengers to downtown and a railway station respectively. 

\begin{figure}[hbt]
  \centering
  \includegraphics[width=4in]{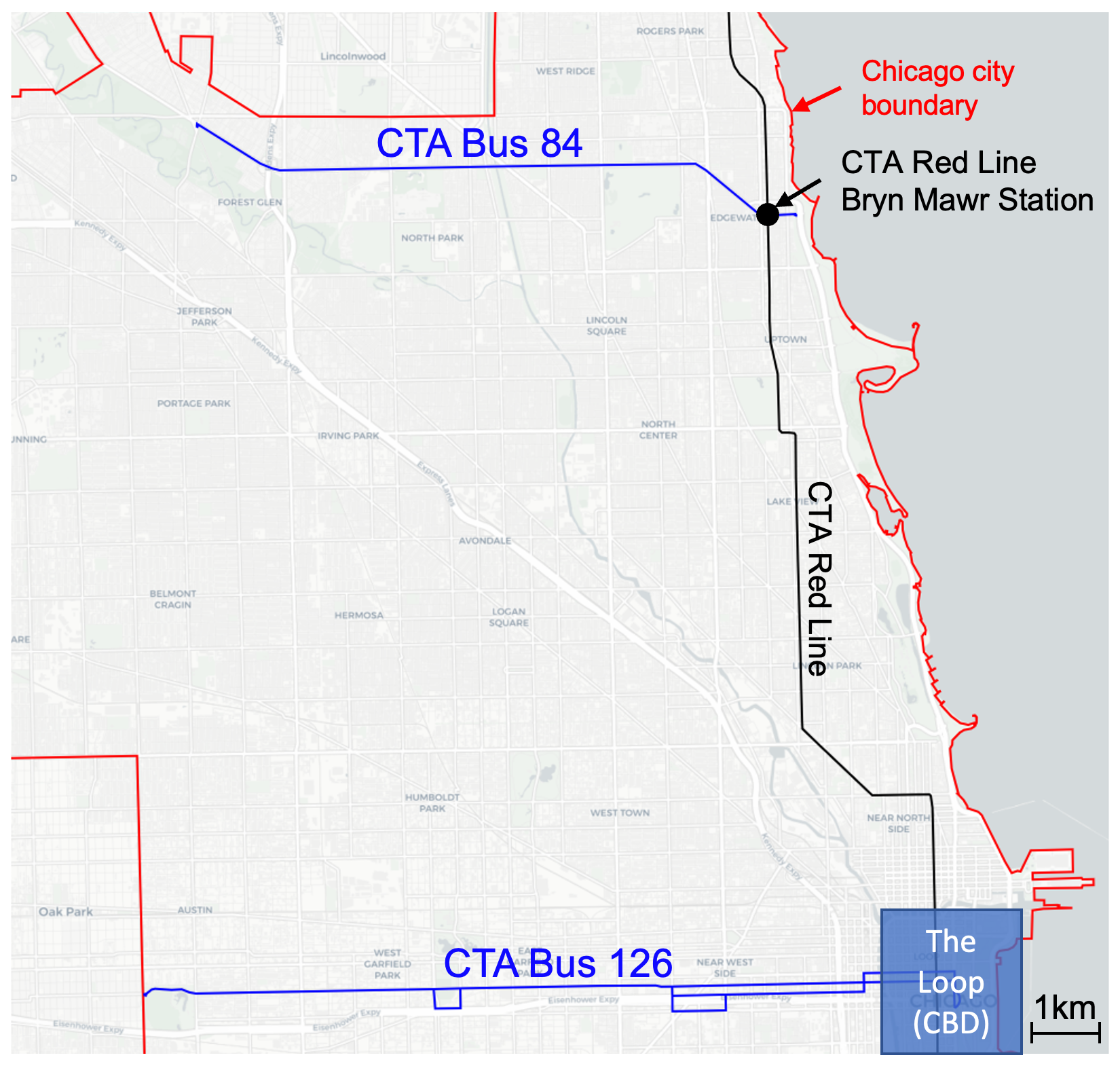}
  \caption{Maps of Bus Routes CTA126 and CTA84 in Chicago \citep{chicago_transit_authority_cta_2022, chicago_transit_authority_cta_2022-1}}
  \label{fig:bus_map}
\end{figure}

The scenario details and parameters are as follows: a demand profile with $\Lambda=80$~pax/h is used with a 15-minute headway $H$ and x-directional demand distributions as uniform and triangular distributions previously discussed in Section~\ref{sec:math_form}. The other parameters are: the value of time $\gamma_t$ is \textdollar16.5/h, cost factors of access(walking) $\gamma_a$ and waiting $\gamma_w$ are 2 and 1.5; the distance-based operational cost $\gamma_o$ is \textdollar0.5/km and vehicle time cost $\gamma_v$ is \textdollar12/h for a minibus \citep{tirachini_economics_2020}; the vehicle speed $V_d$ is 30~km/h and layover time $T_l$ is 10~min; the route length $L_x$ is 10.9~km and 13.4~km for CTA126 and CTA84 respectively; with the assumption of a uniform distribution of y-directional demand in the catchment areas, the average access time $\overline{t_a}$ is 2.25 min for CTA126 and 6.75~min for CTA84 and mean detour $\overline{d_{d,y}}$ is 0.13~km and 0.53~km respectively. 

\subsection{Joint Optimization of the Flexible Route Portion and Fleet Size given a Headway}\label{sec:num_eg_fleet}
The costs and optimal flexible route portion $x_f^*$ are obtained with the analytical formula in Eq.~\ref{eq:c_uni}-\ref{eq:x_f_tri} with the aforementioned parameters.

Under uniform distribution of demand in the x-direction, Figure~\ref{fig:126} shows the total cost, fleet size, and average cost component of the case CTA126, respectively. The optimal flexible route portion $x_{f,uni}^*=7.91$~km serves $F(x_f^*)=62$~pax with the optimal fleet size $s^*=4.76$. When the flexible route portion increases, access cost decreases linearly, while riding cost increases quadratically. As the mean deviation $\overline{d_{d,y}}$ in this case is relatively small, the effect of flexible portions on operating and vehicle costs is minimal, favoring a longer flexible route. Besides, we note that for fixed route operation, the fleet size $s(0)=4.24$, which if not shared across routes would also require 5~veh, the same as the hybrid route service. This suggests that, for scenarios with limited detour impacts on cycle times, the simplified result subject to a fixed fleet size ($\gamma_v = 0$) can provide reasonable initial estimates.
\begin{figure}[hbt]
  \centering
  \includegraphics[width=5.22in]{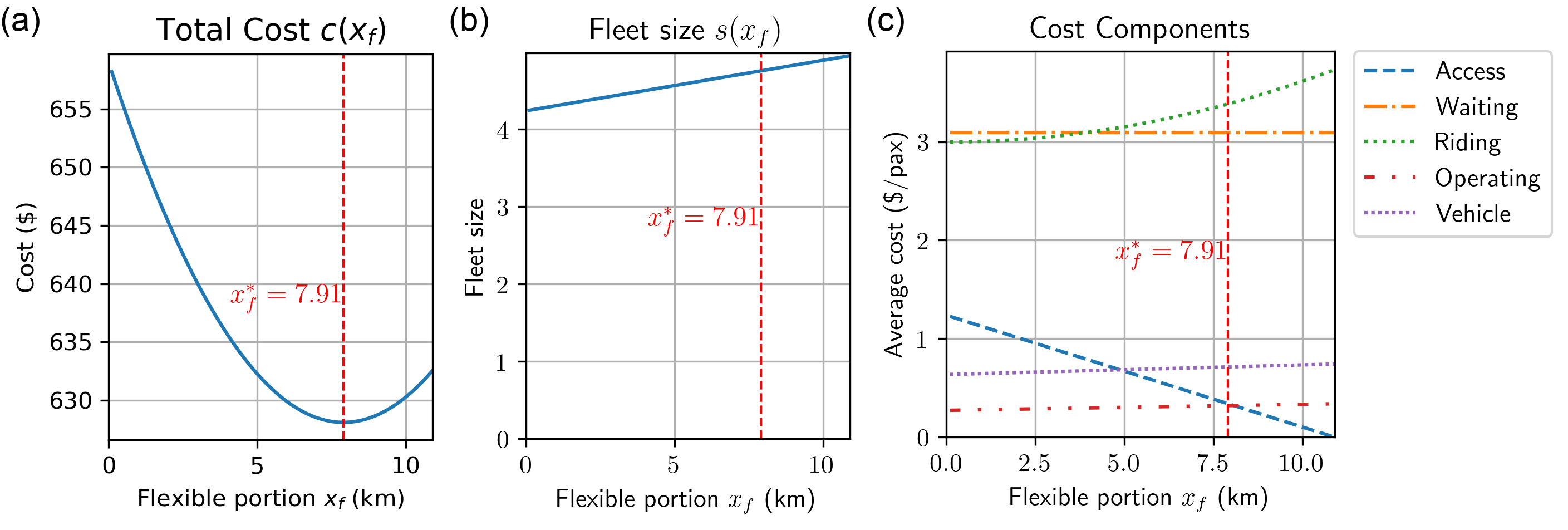}
  \caption{Total Costs $c(x_f)$ (a), Fleet Size $s(x_f)$ (b), and Average Cost Components (c) of Case CTA126 with respect to Flexible Route Portion $x_f$ under Uniform Demand Distribution}
  \label{fig:126}
\end{figure}

The results of triangular distribution in Figure~\ref{fig:126_tri} favor flexible routes even more with $x_{f,tri}^*=9.28$~km. As previously discussed, this is equal to $\sqrt{L_x x_{f,uni}^*}$, or 12.6\% of the route length longer. The concentrated demand closer to the train station favors a longer flexible route portion. On the other hand, access and riding costs change more rapidly closer to the end as shown in Figure~\ref{fig:126_tri}(c). 

\begin{figure}[hbt]
  \centering
  \includegraphics[width=5.22in]{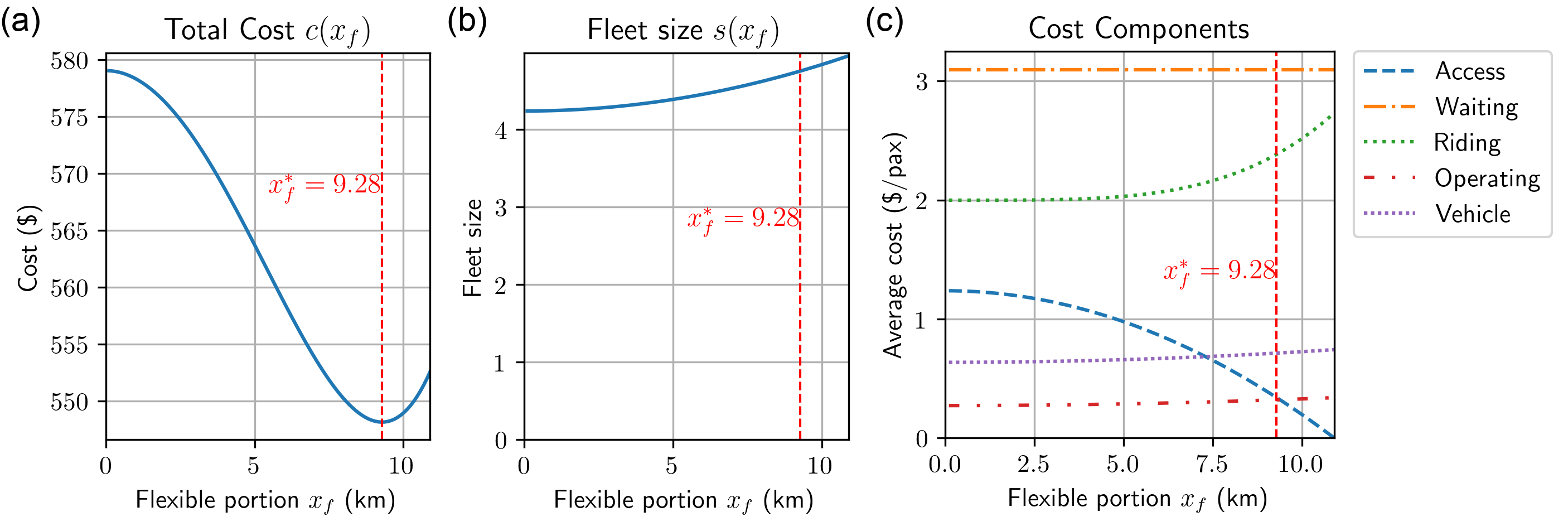}
  \caption{Total Costs $c(x_f)$ (a), Fleet Size $s(x_f)$ (b), and Average Cost Components (c) of Case CTA126 with respect to Flexible Route Portion $x_f$ under Triangular Demand Distribution}
  \label{fig:126_tri}
\end{figure}

Figure~\ref{fig:84} shows similar results for the case CTA84 under a uniform demand distribution. The mean detour $\overline{d_{d,y}}$ and average access time $\overline{t_a}$ are greater than those in the last case, implying higher potential savings in access cost but more vehicle detours for the flexible route portion. The optimal flexible portion $x_{f,uni}^*$ is 6.90~km, serving $F(x_f^*)=45$~pax with the optimal fleet size $s^*=6.37$. The smaller flexible portion is explained by the faster increases in riding, operating, and vehicle costs. 

\begin{figure}[hbt]
  \centering
  \includegraphics[width=5.22in]{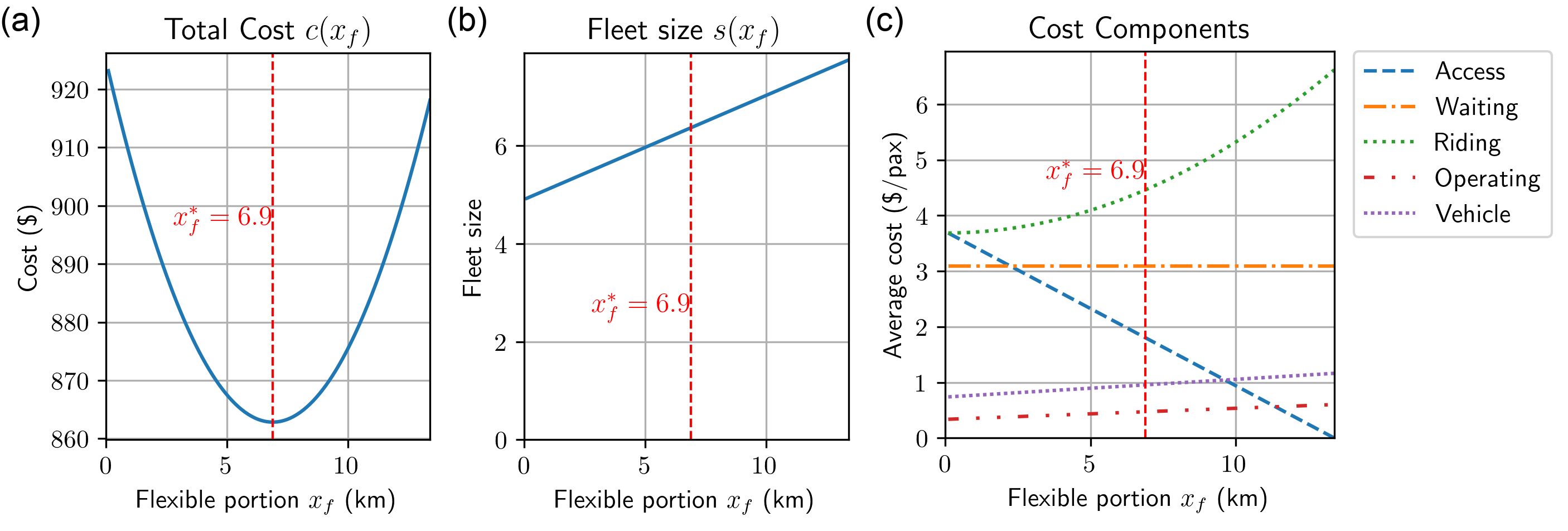}
  \caption{Total Costs $c(x_f)$ (a), Fleet Size $s(x_f)$ (b), and Average Cost Components (c) of Case CTA84 with respect to Flexible Route Portion $x_f$ under Uniform Demand Distributions}
  \label{fig:84}
\end{figure}

The case of a triangular distribution in demand is illustrated in Figure~\ref{fig:84_tri}. While the optimal flexible portion is extended to $x_{f,tri}^*=9.61$~km, the flexible route proportion over the route length is still smaller than CTA126 (CTA84 is a longer route).

\begin{figure}[hbt]
  \centering
  \includegraphics[width=5.22in]{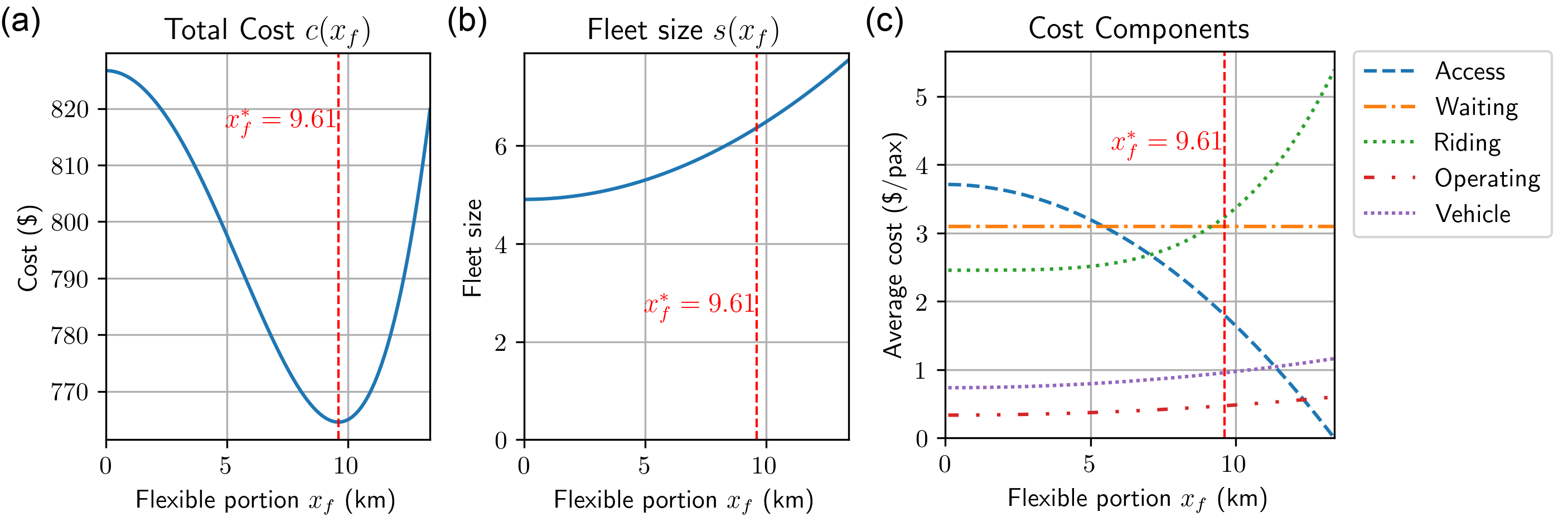}
  \caption{Total Costs $c(x_f)$ (a), Fleet Size $s(x_f)$ (b), and Average Cost Components (c) of Case CTA84 with respect to Flexible Route Portion $x_f$ under Triangular Demand Distributions}
  \label{fig:84_tri}
\end{figure}

In summary, serving passengers further away with a flexible route portion lowers the total cost, in particular user cost, for both routes. Smaller mean detour $\overline{d_{d,y}}$ and an increasing demand gradient favor longer flexible portions.

\subsection{Joint Optimization of the Flexible Route Portion, Fleet Size, and Headway with Variable Vehicle Sizes}
The total cost minimization in Eq.~\ref{eq:c_veh_size} can be achieved by solving for the optimal $x_f$ and $s$ for each vehicle size $b$ with numerical solvers, such as L-BFGS-B \citep{liu_limited_1989} used in this example. This subsection shows results under a uniform distribution of demand. The vehicle types considered are car, van, 20-seater, minibus, and bus, with capacities $b\in$[5,8,20,44,70] (pax/veh), operating costs $\gamma_o\in$ [0.6187, 0.6370, 0.6938, 0.7507, 0.8900] (\$/km) and vehicle costs $\gamma_v\in$ [2.53, 3.63, 7.59, 11.55, 15.73] (\$/h).\footnote{Distance-based vehicle operating costs $\gamma_o$ and time-based capital costs $\gamma_v$ are derived from different vehicle classes and capacities with reference to \citet{tirachini_economics_2020} under their scenario of 50\% reduction in driving costs.} The capacity buffer is set as $\rho=0.7$.

Figure~\ref{fig:126_varH}(a) shows the total costs $c(x_f,s^*,b)$ under the optimal fleet size for the CTA126 case. It illustrates the impact of varying the flexible route portion $x_f$ while optimizing the fleet size. For all vehicle types, the cost curve is monotonically decreasing, indicating only flexible route services. 
For larger vehicles (e.g., standard buses), apart from the larger headway provided, a longer flexible portion implies serving a large number of passengers per trip. The detours for passengers picked up early in the route accumulate, leading to a rapid increase in in-vehicle riding cost ($c_{t,y}$) for all passengers on board. 
Conversely, smaller vehicles (e.g., cars and vans) have lower capacities apart from smaller headway. They complete the flexible pickup phase faster with fewer total deviations, resulting in a flatter cost curve as the flexible portion extends. This explains why smaller vehicles favor fully flexible or highly hybrid routes compared to larger buses. Furthermore, compared to the minimum cost of minibus with fixed headway in Figure~\ref{fig:126}(a), the minimum cost when headway is also optimized now drops to 572 from 627 by more than 8\%. This demonstrates the benefit of the simultaneous optimization.

Figure~\ref{fig:126_varH}(b) shows the optimal fleet sizes $s^*$ under varying flexible route portions $x_f$. To ensure sufficient capacity to meet the demand, the fleet sizes for cars and vans are the minimum required independent of $x_f$. For other vehicle types, the optimal fleet sizes increase gradually to serve more detours. The headways shown in Figure~\ref{fig:126_varH}(c) are much lower than the 15-minute headway setting in the previous examples, suggesting that the reduction in waiting time outweighs the increase in operating cost (Figure~\ref{fig:126_varH}(d)), if the capital and operating costs of SAVs are as low as forecasted (relative to the assumed values of time).

\begin{figure}[hbt]
  \centering
  \includegraphics[width=6in]{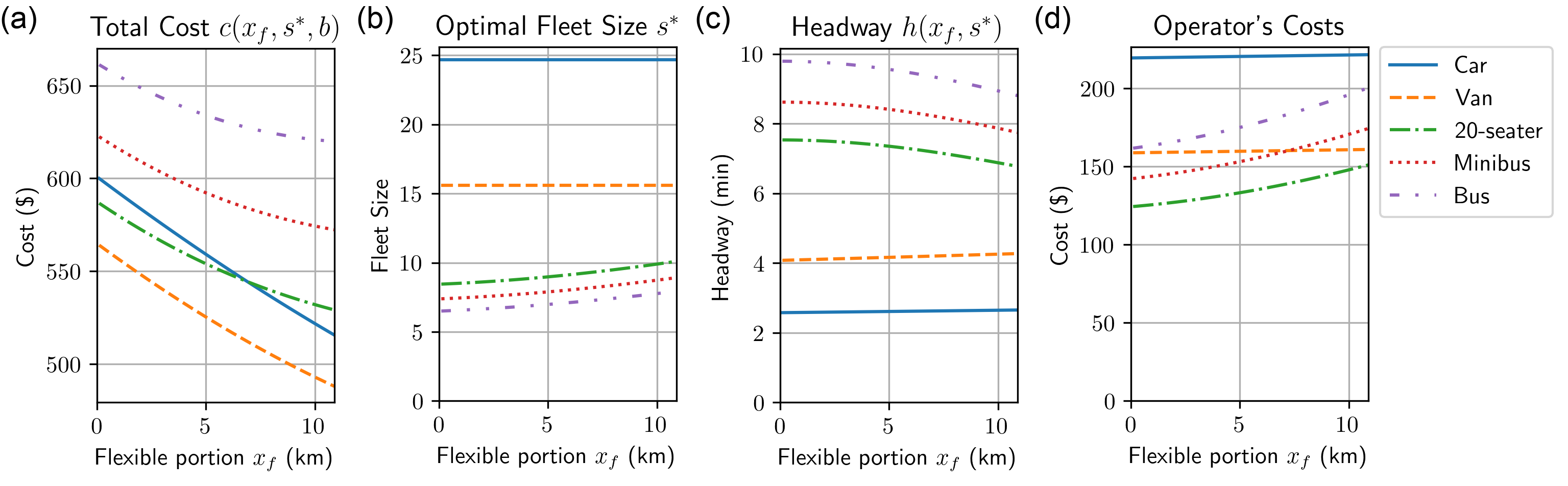}
  \caption{Total Costs under Optimal Fleet Size $c(x_f,s^*,b)$ (a), Optimal Fleet Size $s^*$ (b), Headway $h(x_f,s^*)$ (c), and Total Operator Costs (d) with respect to Flexible Route Portion $x_f$ in Case CTA126}
  \label{fig:126_varH}
\end{figure}

Figure~\ref{fig:84_varH} shows similar results for the case CTA84. The changes with increasing flexible route portions are more considerable due to the larger detours for each passenger.

\begin{figure}[hbt]
  \centering
  \includegraphics[width=6in]{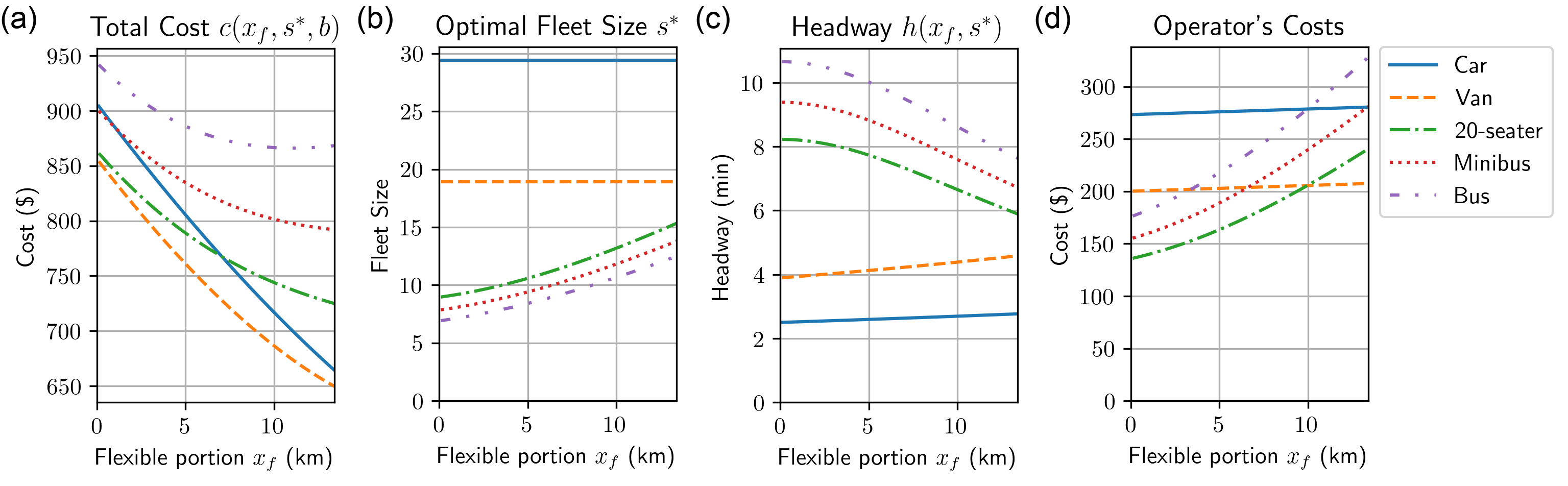}
  \caption{Total Costs under Optimal Fleet Size $c(x_f,s^*,b)$ (a), Optimal Fleet Size $s^*$ (b), Headway $h(x_f,s^*)$ (c), and Total Operator Costs (d) with respect to Flexible Route Portion $x_f$ in Case CTA84}
  \label{fig:84_varH}
\end{figure}

Flexible routes are optimal in most cases, different from the previous results in Section~\ref{sec:num_eg_fleet} where hybrid routes are optimal under a constant 15-minute headway. This is primarily contributed by the much lower headway, particularly for smaller vehicle sizes. However, the operator costs required to provide such low headway are considerably higher than the current human-driven fixed route operations (\$124 and \$150 for CTA126 and CTA84 respectively). If a budget constraint is applied, the optimal results would likely be hybrid route similar to Section~\ref{sec:num_eg_fleet} where changes in operator costs are much smaller under fixed headway.

To further disentangle the impacts of the primary decision variables, Figure~\ref{fig:3d_surfaces} visualizes the total generalized cost surfaces with respect to the flexible route portion $x_f$ and fleet size $s$. It is important to note that headway $h$ and $s$ are directly coupled through the cycle time relationship (Eq.~\ref{eq:s_fleet}), so we can simplify the analysis for each vehicle size $b$ to optimizing $x_f$ and $s$ with $h$ determined endogenously. The results reveal that the sensitivity to these variables is fundamentally determined by the vehicle size. For smaller vehicles (e.g., car in Figure~\ref{fig:3d_surfaces}), the feasible region is limited to high $s$ to satisfy capacity constraint. The cost surface shows the global minimum at a high flexible route portion and minimum fleet size. The steep gradients along both axes indicate that the system is highly sensitive to sub-optimal settings in either variable. Conversely, for large vehicles (e.g., minibus and bus in Figure~\ref{fig:3d_surfaces}), the optimal solution shifts to a purely fixed route ($x_f = 0$) at small fleet size (cost curve increases monotonically along $x_f$). The cost surface is also convex along the $s$ axis, suggesting the importance of a balance between increasing vehicles for a smaller headway and limit spending on operating costs. This demonstrates that vehicle size acts as the critical enabling variable: smaller vehicles unlock the significant potential of optimizing the flexible route portion alongside fleet size, whereas for traditional large transit vehicles, fleet size remains the dominant tactical decision variable before hybrid/flexible routes are attractive.

Table~\ref{tab:veh_size} lists the detailed results. 

\begin{landscape}

\begin{figure}[hbt]
  \centering
  \includegraphics[width=8in]{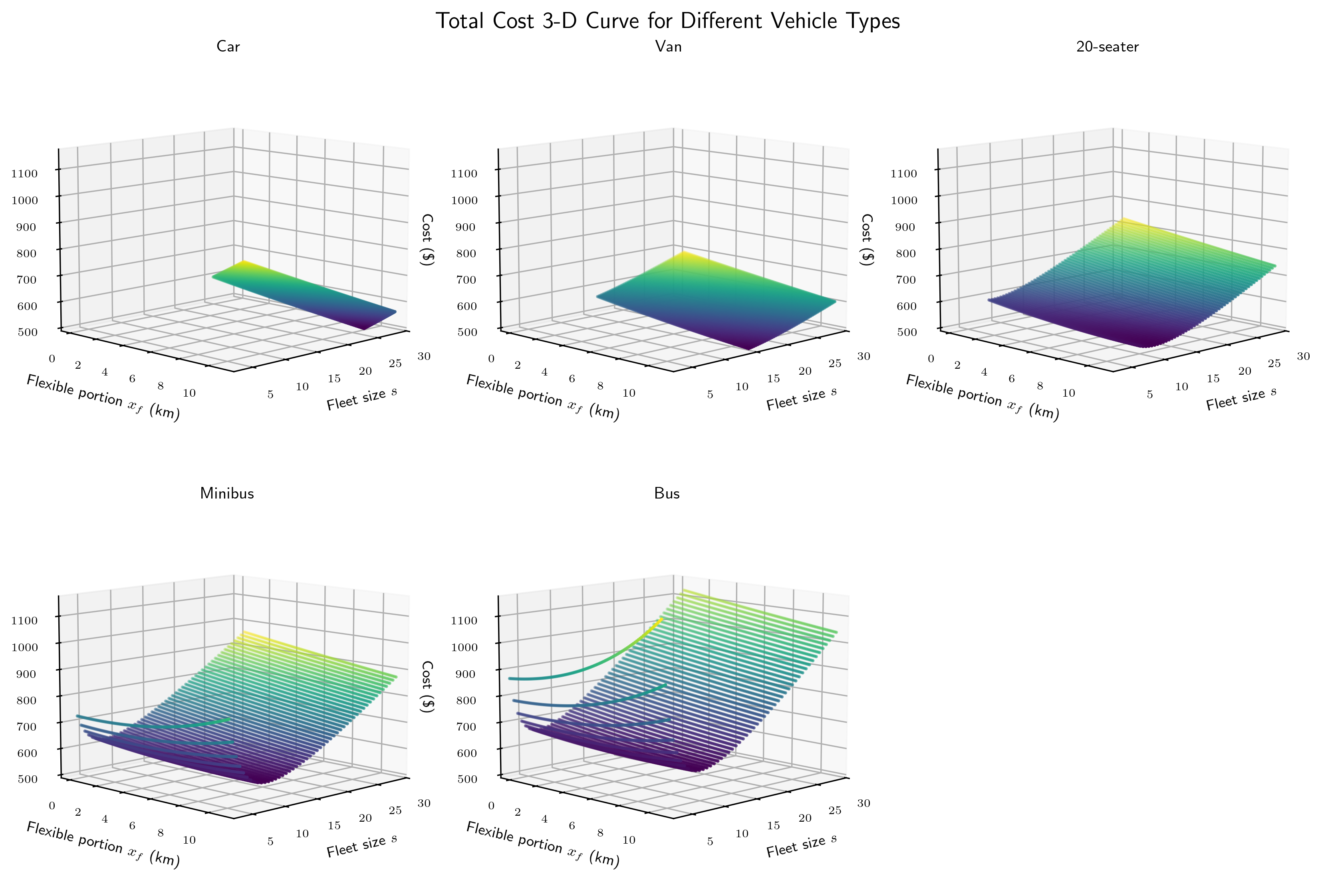}
  \caption{Total Cost 3-D Curves under Varying Flexible Route Portions and Fleet Sizes $c(x_f,s,b)$ for Different Vehicle Sizes in Case CTA84}
  \label{fig:3d_surfaces}
\end{figure}

\begin{table}[]
\caption{Results of Joint Optimization of Vehicle Size, Flexible Route Portion, and Fleet Size}
\label{tab:veh_size}
\begin{tabular}{|clll|ccccc|ccccc|}
\hline
\multicolumn{4}{|l|}{\textbf{Bus route}}                                                                                                                                                                                                                                                               & \multicolumn{5}{c|}{\textbf{CTA126}}                                                                                                                                                                                                                                                                                     & \multicolumn{5}{c|}{\textbf{CTA84}}                                                                                                                                                                                                                                                                                      \\ \hline
\multicolumn{1}{|c|}{\multirow{2}{*}{\textbf{Vehicle}}}                                                    & \multicolumn{3}{l|}{\textbf{Type}}                                                                                                                                                        & \multicolumn{1}{c|}{\textbf{Car}}   & \multicolumn{1}{c|}{\textbf{Van}}   & \multicolumn{1}{c|}{\textbf{\begin{tabular}[c]{@{}c@{}}20-\\ Seater\end{tabular}}} & \multicolumn{1}{c|}{\textbf{\begin{tabular}[c]{@{}c@{}}Mini-\\ bus\end{tabular}}} & \textbf{\begin{tabular}[c]{@{}c@{}}Stand-\\ ard\\ bus\end{tabular}} & \multicolumn{1}{c|}{\textbf{Car}}   & \multicolumn{1}{c|}{\textbf{Van}}   & \multicolumn{1}{c|}{\textbf{\begin{tabular}[c]{@{}c@{}}20-\\ Seater\end{tabular}}} & \multicolumn{1}{c|}{\textbf{\begin{tabular}[c]{@{}c@{}}Mini-\\ bus\end{tabular}}} & \textbf{\begin{tabular}[c]{@{}c@{}}Stand-\\ ard\\ bus\end{tabular}} \\ \cline{2-14} 
\multicolumn{1}{|c|}{}                                                                                     & \multicolumn{1}{l|}{\textbf{Size}}                                                           & \multicolumn{1}{l|}{b}                & \begin{tabular}[c]{@{}l@{}}pax\\ /veh\end{tabular} & \multicolumn{1}{c|}{5}              & \multicolumn{1}{c|}{8}              & \multicolumn{1}{c|}{20}                                                            & \multicolumn{1}{c|}{44}                                                           & 70                                                                  & \multicolumn{1}{c|}{5}              & \multicolumn{1}{c|}{8}              & \multicolumn{1}{c|}{20}                                                            & \multicolumn{1}{c|}{44}                                                           & 70                                                                  \\ \hline
\multicolumn{1}{|c|}{\multirow{3}{*}{\textbf{\begin{tabular}[c]{@{}c@{}}Optimal\\ variable\end{tabular}}}} & \multicolumn{1}{l|}{\textbf{\begin{tabular}[c]{@{}l@{}}Flexible\\ portion\end{tabular}}}     & \multicolumn{1}{l|}{$x_f^*$}          & km                                                 & \multicolumn{1}{c|}{\textit{10.90}} & \multicolumn{1}{c|}{\textit{10.90}} & \multicolumn{1}{c|}{\textit{10.90}}                                                & \multicolumn{1}{c|}{\textit{10.90}}                                               & \textit{10.90}                                                      & \multicolumn{1}{c|}{\textit{13.40}} & \multicolumn{1}{c|}{\textit{13.40}} & \multicolumn{1}{c|}{\textit{13.40}}                                                & \multicolumn{1}{c|}{\textit{13.40}}                                               & 11.04                                                               \\ \cline{2-14} 
\multicolumn{1}{|c|}{}                                                                                     & \multicolumn{1}{l|}{\textbf{Fleet size}}                                                     & \multicolumn{1}{l|}{$s^*$}            & veh                                                & \multicolumn{1}{c|}{\textit{24.69}} & \multicolumn{1}{c|}{\textit{15.60}} & \multicolumn{1}{c|}{10.10}                                                         & \multicolumn{1}{c|}{8.92}                                                         & 7.93                                                                & \multicolumn{1}{c|}{\textit{29.42}} & \multicolumn{1}{c|}{\textit{18.90}} & \multicolumn{1}{c|}{15.33}                                                         & \multicolumn{1}{c|}{13.79}                                                        & 11.20                                                               \\ \cline{2-14} 
\multicolumn{1}{|c|}{}                                                                                     & \multicolumn{1}{l|}{\textbf{Headway}}                                                        & \multicolumn{1}{l|}{$h^*$}            & min                                                & \multicolumn{1}{c|}{2.65}           & \multicolumn{1}{c|}{4.27}           & \multicolumn{1}{c|}{6.77}                                                          & \multicolumn{1}{c|}{7.75}                                                         & 8.81                                                                & \multicolumn{1}{c|}{2.77}           & \multicolumn{1}{c|}{4.58}           & \multicolumn{1}{c|}{5.89}                                                          & \multicolumn{1}{c|}{6.72}                                                         & 8.31                                                                \\ \hline
\multicolumn{1}{|c|}{\multirow{3}{*}{\textbf{\begin{tabular}[c]{@{}c@{}}Average\\ time\end{tabular}}}}     & \multicolumn{1}{l|}{\textbf{Access}}                                                         & \multicolumn{1}{l|}{$\overline{t_a}$} & min                                                & \multicolumn{1}{c|}{0.00}           & \multicolumn{1}{c|}{0.00}           & \multicolumn{1}{c|}{0.00}                                                          & \multicolumn{1}{c|}{0.00}                                                         & 0.00                                                                & \multicolumn{1}{c|}{0.00}           & \multicolumn{1}{c|}{0.00}           & \multicolumn{1}{c|}{0.00}                                                          & \multicolumn{1}{c|}{0.00}                                                         & 1.19                                                                \\ \cline{2-14} 
\multicolumn{1}{|c|}{}                                                                                     & \multicolumn{1}{l|}{\textbf{Waiting}}                                                        & \multicolumn{1}{l|}{$\overline{t_w}$} & min                                                & \multicolumn{1}{c|}{1.33}           & \multicolumn{1}{c|}{2.14}           & \multicolumn{1}{c|}{3.39}                                                          & \multicolumn{1}{c|}{3.87}                                                         & 4.40                                                                & \multicolumn{1}{c|}{1.38}           & \multicolumn{1}{c|}{2.29}           & \multicolumn{1}{c|}{2.95}                                                          & \multicolumn{1}{c|}{3.36}                                                         & 4.16                                                                \\ \cline{2-14} 
\multicolumn{1}{|c|}{}                                                                                     & \multicolumn{1}{l|}{\textbf{Riding}}                                                         & \multicolumn{1}{l|}{$\overline{t_t}$} & min                                                & \multicolumn{1}{c|}{11.37}          & \multicolumn{1}{c|}{11.66}          & \multicolumn{1}{c|}{12.10}                                                         & \multicolumn{1}{c|}{12.28}                                                        & 12.47                                                               & \multicolumn{1}{c|}{15.37}          & \multicolumn{1}{c|}{16.66}          & \multicolumn{1}{c|}{17.59}                                                         & \multicolumn{1}{c|}{18.18}                                                        & 17.41                                                               \\ \hline
\multicolumn{1}{|c|}{\multirow{3}{*}{\textbf{\begin{tabular}[c]{@{}c@{}}Time\\ std. dev.\end{tabular}}}}   & \multicolumn{1}{l|}{\textbf{Access}}                                                         & \multicolumn{1}{l|}{$\sigma_{t,a}$}   & min                                                & \multicolumn{1}{c|}{0.00}           & \multicolumn{1}{c|}{0.00}           & \multicolumn{1}{c|}{0.00}                                                          & \multicolumn{1}{c|}{0.00}                                                         & 0.00                                                                & \multicolumn{1}{c|}{0.00}           & \multicolumn{1}{c|}{0.00}           & \multicolumn{1}{c|}{0.00}                                                          & \multicolumn{1}{c|}{0.00}                                                         & 1.45                                                                \\ \cline{2-14} 
\multicolumn{1}{|c|}{}                                                                                     & \multicolumn{1}{l|}{\textbf{Waiting}}                                                        & \multicolumn{1}{l|}{$\sigma_{t,w}$}   & min                                                & \multicolumn{1}{c|}{0.77}           & \multicolumn{1}{c|}{1.23}           & \multicolumn{1}{c|}{1.96}                                                          & \multicolumn{1}{c|}{2.24}                                                         & 2.54                                                                & \multicolumn{1}{c|}{0.80}           & \multicolumn{1}{c|}{1.32}           & \multicolumn{1}{c|}{1.70}                                                          & \multicolumn{1}{c|}{1.94}                                                         & 2.40                                                                \\ \cline{2-14} 
\multicolumn{1}{|c|}{}                                                                                     & \multicolumn{1}{l|}{\textbf{Riding}}                                                         & \multicolumn{1}{l|}{$\sigma_{t,r}$}   & min                                                & \multicolumn{1}{c|}{6.75}           & \multicolumn{1}{c|}{6.93}           & \multicolumn{1}{c|}{7.20}                                                          & \multicolumn{1}{c|}{7.30}                                                         & 7.41                                                                & \multicolumn{1}{c|}{9.63}           & \multicolumn{1}{c|}{10.42}          & \multicolumn{1}{c|}{10.98}                                                         & \multicolumn{1}{c|}{11.33}                                                        & 11.04                                                               \\ \hline
\multicolumn{1}{|c|}{\multirow{8}{*}{\textbf{Cost}}}                                                       & \multicolumn{1}{l|}{\textbf{Access}}                                                         & \multicolumn{1}{l|}{$c_a$}            & \$                                                 & \multicolumn{1}{c|}{0.00}           & \multicolumn{1}{c|}{0.00}           & \multicolumn{1}{c|}{0.00}                                                          & \multicolumn{1}{c|}{0.00}                                                         & 0.00                                                                & \multicolumn{1}{c|}{0.00}           & \multicolumn{1}{c|}{0.00}           & \multicolumn{1}{c|}{0.00}                                                          & \multicolumn{1}{c|}{0.00}                                                         & 52.24                                                               \\ \cline{2-14} 
\multicolumn{1}{|c|}{}                                                                                     & \multicolumn{1}{l|}{\textbf{Waiting}}                                                        & \multicolumn{1}{l|}{$c_w$}            & \$                                                 & \multicolumn{1}{c|}{43.76}          & \multicolumn{1}{c|}{70.48}          & \multicolumn{1}{c|}{111.76}                                                        & \multicolumn{1}{c|}{127.85}                                                       & 145.29                                                              & \multicolumn{1}{c|}{45.70}          & \multicolumn{1}{c|}{75.64}          & \multicolumn{1}{c|}{97.24}                                                         & \multicolumn{1}{c|}{110.96}                                                       & 137.12                                                              \\ \cline{2-14} 
\multicolumn{1}{|c|}{}                                                                                     & \multicolumn{1}{l|}{\textbf{\begin{tabular}[c]{@{}l@{}}Riding\\ (x-dir.)\end{tabular}}}      & \multicolumn{1}{l|}{$c_{t,x}$}        & \$                                                 & \multicolumn{1}{c|}{239.80}         & \multicolumn{1}{c|}{239.80}         & \multicolumn{1}{c|}{239.80}                                                        & \multicolumn{1}{c|}{239.80}                                                       & 239.80                                                              & \multicolumn{1}{c|}{294.80}         & \multicolumn{1}{c|}{294.80}         & \multicolumn{1}{c|}{294.80}                                                        & \multicolumn{1}{c|}{294.80}                                                       & 294.80                                                              \\ \cline{2-14} 
\multicolumn{1}{|c|}{}                                                                                     & \multicolumn{1}{l|}{\textbf{\begin{tabular}[c]{@{}l@{}}Riding\\ (y-dir.)\end{tabular}}}      & \multicolumn{1}{l|}{$c_{t,y}$}        & \$                                                 & \multicolumn{1}{c|}{10.37}          & \multicolumn{1}{c|}{16.71}          & \multicolumn{1}{c|}{26.49}                                                         & \multicolumn{1}{c|}{30.31}                                                        & 34.44                                                               & \multicolumn{1}{c|}{43.33}          & \multicolumn{1}{c|}{71.72}          & \multicolumn{1}{c|}{92.20}                                                         & \multicolumn{1}{c|}{105.20}                                                       & 88.30                                                               \\ \cline{2-14} 
\multicolumn{1}{|c|}{}                                                                                     & \multicolumn{1}{l|}{\textbf{\begin{tabular}[c]{@{}l@{}}Operational\\ (x-dir.)\end{tabular}}} & \multicolumn{1}{l|}{$c_{o,x}$}        & \$                                                 & \multicolumn{1}{c|}{152.56}         & \multicolumn{1}{c|}{97.53}          & \multicolumn{1}{c|}{66.99}                                                         & \multicolumn{1}{c|}{63.36}                                                        & 66.10                                                               & \multicolumn{1}{c|}{179.61}         & \multicolumn{1}{c|}{111.72}         & \multicolumn{1}{c|}{94.65}                                                         & \multicolumn{1}{c|}{89.75}                                                        & 86.11                                                               \\ \cline{2-14} 
\multicolumn{1}{|c|}{}                                                                                     & \multicolumn{1}{l|}{\textbf{\begin{tabular}[c]{@{}l@{}}Operational\\ (y-dir.)\end{tabular}}} & \multicolumn{1}{l|}{$c_{o,y}$}        & \$                                                 & \multicolumn{1}{c|}{6.60}           & \multicolumn{1}{c|}{6.79}           & \multicolumn{1}{c|}{7.40}                                                          & \multicolumn{1}{c|}{8.01}                                                         & 9.49                                                                & \multicolumn{1}{c|}{26.40}          & \multicolumn{1}{c|}{27.18}          & \multicolumn{1}{c|}{29.60}                                                         & \multicolumn{1}{c|}{32.03}                                                        & 31.29                                                               \\ \cline{2-14} 
\multicolumn{1}{|c|}{}                                                                                     & \multicolumn{1}{l|}{\textbf{Vehicle}}                                                        & \multicolumn{1}{l|}{$c_v$}            & \$                                                 & \multicolumn{1}{c|}{62.47}          & \multicolumn{1}{c|}{56.63}          & \multicolumn{1}{c|}{76.66}                                                         & \multicolumn{1}{c|}{103.01}                                                       & 124.80                                                              & \multicolumn{1}{c|}{74.43}          & \multicolumn{1}{c|}{68.61}          & \multicolumn{1}{c|}{116.38}                                                        & \multicolumn{1}{c|}{159.26}                                                       & 176.19                                                              \\ \cline{2-14} 
\multicolumn{1}{|c|}{}                                                                                     & \multicolumn{1}{l|}{\textbf{Total}}                                                          & \multicolumn{1}{l|}{$c$}              & \$                                                 & \multicolumn{1}{c|}{515.56}         & \multicolumn{1}{c|}{487.94}         & \multicolumn{1}{c|}{529.11}                                                        & \multicolumn{1}{c|}{572.34}                                                       & 619.93                                                              & \multicolumn{1}{c|}{664.27}         & \multicolumn{1}{c|}{649.66}         & \multicolumn{1}{c|}{724.87}                                                        & \multicolumn{1}{c|}{792.01}                                                       & 866.04                                                              \\ \hline
\end{tabular}
\end{table}
\end{landscape}

\subsubsection{Sensitivity to Operator Cost}

To assess the sensitivity to the forecasted operator costs (or user value of time), Figures~\ref{fig:sen_VOT} and \ref{fig:sen_VOT_300} present the results under 200\% and 300\% of operator costs (or effectively 50\% and 33\% of users' value of time) for case CTA84. These scenarios serve as a proxy for comparing SAV operations (baseline low cost) against lower automation or human-driven operations (higher cost).

\begin{figure}[hbt]
  \centering
  \includegraphics[width=6in]{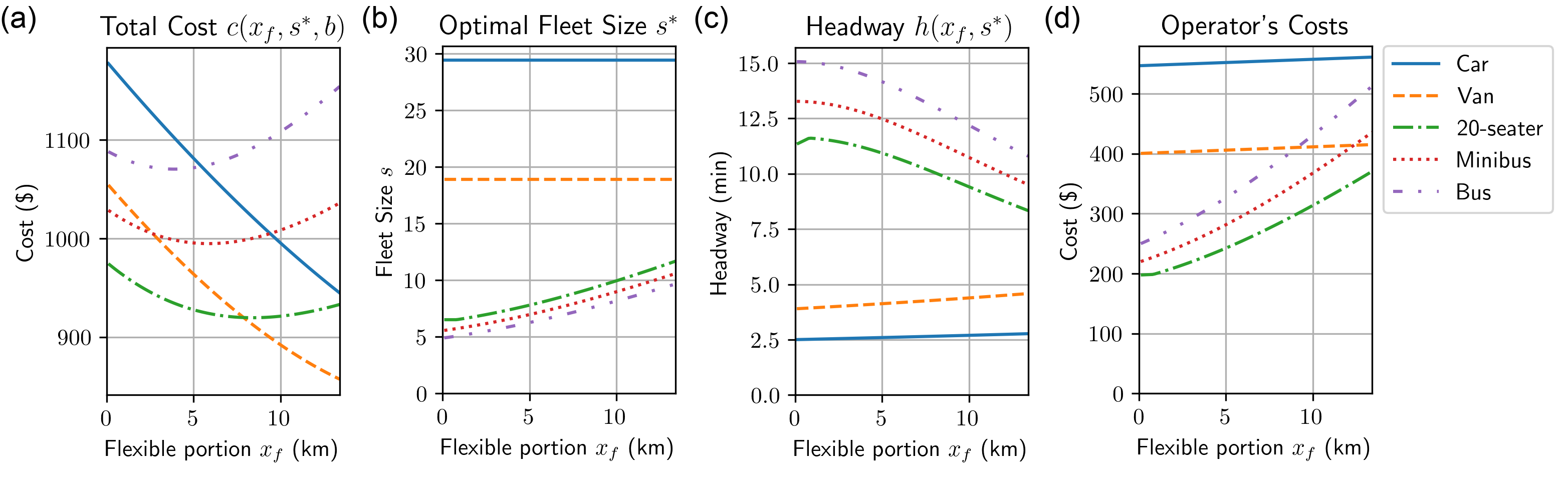}
  \caption{Total Costs under Optimal Fleet Size $c(x_f,s^*,b)$ (a), Optimal Fleet Size $s^*$ (b), Headway $h(x_f,s^*)$ (c), and Total Operator Costs (d) with respect to Flexible Route Portion $x_f$ under 200\% Operator Costs for Case CTA84}
  \label{fig:sen_VOT}
\end{figure}

\begin{figure}[hbt]
  \centering
  \includegraphics[width=6in]{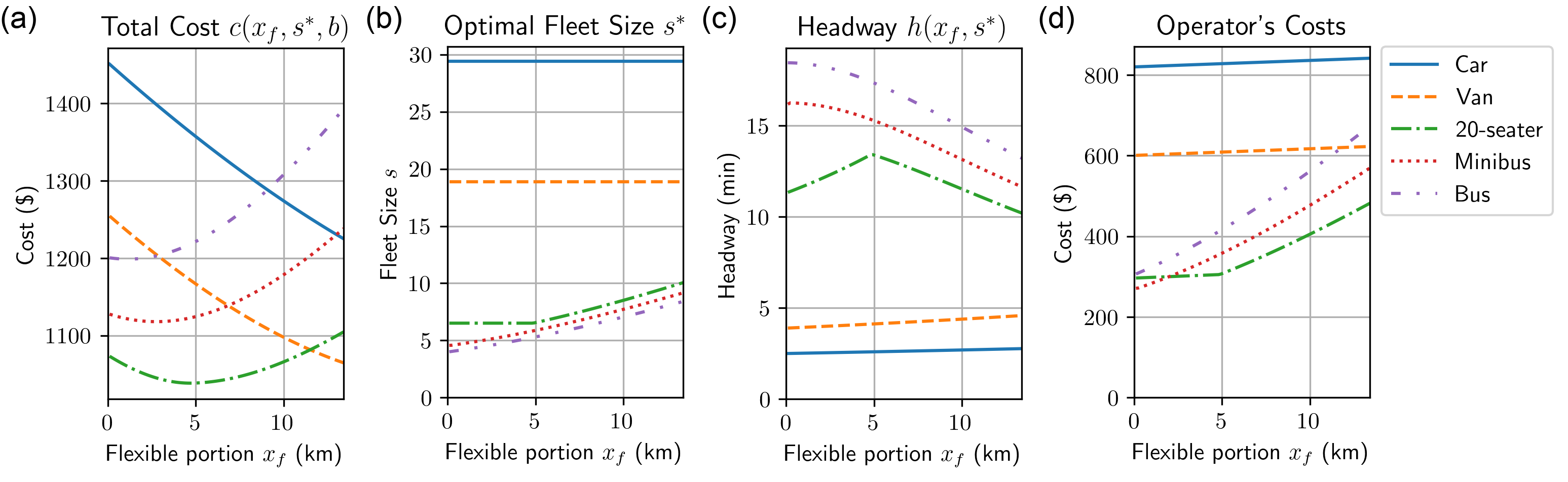}
  \caption{Total Costs under Optimal Fleet Size $c(x_f,s^*,b)$ (a), Optimal Fleet Size $s^*$ (b), Headway $h(x_f,s^*)$ (c), and Total Operator Costs (d) with respect to Flexible Route Portion $x_f$ under 300\% Operator Costs for Case CTA84}
  \label{fig:sen_VOT_300}
\end{figure}

As operating costs increase, we observe a distinct reduction in the optimal flexible route portion in subplots (a). Particularly, in the 300\% cost scenario shown in Figure~\ref{fig:sen_VOT_300}, the optimal solution (the cost minimum point) shifts away from fully flexible operations, with the hybrid route becoming dominant for bigger vehicles (20-seater, minibus, and bus). This confirms that the economic viability of extensive flexible routing is heavily dependent on keeping operating costs low by SAVs.

Additionally, the optimal flexible route portion decreases with increasing vehicle sizes. This finding implies that smaller vehicles, despite requiring larger fleet sizes, can better support longer flexible route portions. Each trip with a smaller vehicle serves fewer passengers and incurs fewer detours, contrasting with larger vehicles where the accumulation of detours per trip can become excessive. Although more smaller vehicles are required to serve the same demand, each trip involves smaller detours in the flexible route portion alongside shorter waiting times. This also explains the general reduction in headway with an increase in the flexible route portion in subplots (c). These findings align with Result~\ref{thm:route_choice}, where lower operational costs ($\gamma_o$ and $\gamma_v$) and smaller headway ($H$) favor flexible routes. Specifically for the case of CTA84, we can draw theoretical bounds on $\gamma_o$ that favor fixed routes as $\gamma_o \geq 7 - 0.067 \gamma_v$ from Result~\ref{thm:route_choice} (with constant headway assumed). Substituting the operating costs of different vehicle types suggests, even at conservative cost estimates (300\%) for high-capacity vehicles, hybrid routes still bring some benefits. However, their extent would be greatly limited by operating costs following Eq.~\ref{eq:F_gen}. For example, the optimal flexible route portion for a minibus would drop from fully flexible (13.4~km) at 100\% cost to 5.87~km at 200\% cost and 2.75~km at 300\% cost.

In short, a balance across vehicle size is important --- using 4-passenger cars would necessitate a lot of vehicles to serve the demand; utilizing larger buses could result in excessive detours and longer waiting times. The optimal scenario with the assumed operator costs brought by SAVs in both cases is using a van of 8 passengers with only flexible route service. However, until SAV technology matures and operating costs decrease significantly, larger vehicles operating hybrid routes might represent a more pragmatic or economically optimal interim solution. This also applies to cases where operators face a strict budget such that the system-optimally low headway is infeasible. Table~\ref{tab:veh_size_sen} shows the detailed results.

\begin{landscape}
\begin{table}[]
\caption{Results of Sensitivity Analysis for the Joint Optimization of Vehicle Size, Flexible Route Portion, and Fleet Size}
\label{tab:veh_size_sen}
\begin{tabular}{|clll|ccccc|ccccc|}
\hline
\multicolumn{4}{|l|}{\textbf{Bus route CTA84}}                                                                                                                                                                                                                                                         & \multicolumn{5}{c|}{\textbf{200\% Operator Cost}}                                                                                                                                                                                                                                                                      & \multicolumn{5}{c|}{\textbf{300\% Operator Cost}}                                                                                                                                                                                                                                                                      \\ \hline
\multicolumn{1}{|c|}{\multirow{2}{*}{\textbf{Vehicle}}}                                                    & \multicolumn{3}{l|}{\textbf{Type}}                                                                                                                                                        & \multicolumn{1}{c|}{\textbf{Car}}   & \multicolumn{1}{c|}{\textbf{Van}}   & \multicolumn{1}{c|}{\textbf{\begin{tabular}[c]{@{}c@{}}20-\\ Seater\end{tabular}}} & \multicolumn{1}{c|}{\textbf{\begin{tabular}[c]{@{}c@{}}Mini-\\ bus\end{tabular}}} & \textbf{\begin{tabular}[c]{@{}c@{}}Stand-\\ ard\\ bus\end{tabular}} & \multicolumn{1}{c|}{\textbf{Car}}   & \multicolumn{1}{c|}{\textbf{Van}}   & \multicolumn{1}{c|}{\textbf{\begin{tabular}[c]{@{}c@{}}20-\\ Seater\end{tabular}}} & \multicolumn{1}{c|}{\textbf{\begin{tabular}[c]{@{}c@{}}Mini-\\ bus\end{tabular}}} & \textbf{\begin{tabular}[c]{@{}c@{}}Stand-\\ ard\\ bus\end{tabular}} \\ \cline{2-14} 
\multicolumn{1}{|c|}{}                                                                                     & \multicolumn{1}{l|}{\textbf{Size}}                                                           & \multicolumn{1}{l|}{b}                & \begin{tabular}[c]{@{}l@{}}pax\\ /veh\end{tabular} & \multicolumn{1}{c|}{5}              & \multicolumn{1}{c|}{8}              & \multicolumn{1}{c|}{20}                                                            & \multicolumn{1}{c|}{44}                                                           & 70                                                                  & \multicolumn{1}{c|}{5}              & \multicolumn{1}{c|}{8}              & \multicolumn{1}{c|}{20}                                                            & \multicolumn{1}{c|}{44}                                                           & 70                                                                  \\ \hline
\multicolumn{1}{|c|}{\multirow{3}{*}{\textbf{\begin{tabular}[c]{@{}c@{}}Optimal\\ variable\end{tabular}}}} & \multicolumn{1}{l|}{\textbf{\begin{tabular}[c]{@{}l@{}}Flexible\\ portion\end{tabular}}}     & \multicolumn{1}{l|}{$x_f^*$}          & km                                                 & \multicolumn{1}{c|}{\textit{13.40}} & \multicolumn{1}{c|}{\textit{13.40}} & \multicolumn{1}{c|}{8.39}                                                          & \multicolumn{1}{c|}{5.87}                                                         & 3.89                                                                & \multicolumn{1}{c|}{\textit{13.40}} & \multicolumn{1}{c|}{\textit{13.40}} & \multicolumn{1}{c|}{4.69}                                                          & \multicolumn{1}{c|}{2.75}                                                         & \textit{1.14}                                                       \\ \cline{2-14} 
\multicolumn{1}{|c|}{}                                                                                     & \multicolumn{1}{l|}{\textbf{Fleet size}}                                                     & \multicolumn{1}{l|}{$s^*$}            & veh                                                & \multicolumn{1}{c|}{\textit{29.42}} & \multicolumn{1}{c|}{\textit{18.90}} & \multicolumn{1}{c|}{9.19}                                                          & \multicolumn{1}{c|}{7.27}                                                         & 5.90                                                                & \multicolumn{1}{c|}{\textit{29.42}} & \multicolumn{1}{c|}{\textit{18.90}} & \multicolumn{1}{c|}{6.51}                                                          & \multicolumn{1}{c|}{5.20}                                                         & 4.24                                                                \\ \cline{2-14} 
\multicolumn{1}{|c|}{}                                                                                     & \multicolumn{1}{l|}{\textbf{Headway}}                                                        & \multicolumn{1}{l|}{$h^*$}            & min                                                & \multicolumn{1}{c|}{2.77}           & \multicolumn{1}{c|}{4.58}           & \multicolumn{1}{c|}{9.93}                                                          & \multicolumn{1}{c|}{12.21}                                                        & 14.50                                                               & \multicolumn{1}{c|}{2.77}           & \multicolumn{1}{c|}{4.58}           & \multicolumn{1}{c|}{13.35}                                                         & \multicolumn{1}{c|}{15.94}                                                        & 18.39                                                               \\ \hline
\multicolumn{1}{|c|}{\multirow{3}{*}{\textbf{\begin{tabular}[c]{@{}c@{}}Average\\ time\end{tabular}}}}     & \multicolumn{1}{l|}{\textbf{Access}}                                                         & \multicolumn{1}{l|}{$\overline{t_a}$} & min                                                & \multicolumn{1}{c|}{0.00}           & \multicolumn{1}{c|}{0.00}           & \multicolumn{1}{c|}{2.52}                                                          & \multicolumn{1}{c|}{3.80}                                                         & 4.79                                                                & \multicolumn{1}{c|}{0.00}           & \multicolumn{1}{c|}{0.00}           & \multicolumn{1}{c|}{4.39}                                                          & \multicolumn{1}{c|}{5.37}                                                         & 6.18                                                                \\ \cline{2-14} 
\multicolumn{1}{|c|}{}                                                                                     & \multicolumn{1}{l|}{\textbf{Waiting}}                                                        & \multicolumn{1}{l|}{$\overline{t_w}$} & min                                                & \multicolumn{1}{c|}{1.38}           & \multicolumn{1}{c|}{2.29}           & \multicolumn{1}{c|}{4.97}                                                          & \multicolumn{1}{c|}{6.11}                                                         & 7.25                                                                & \multicolumn{1}{c|}{1.38}           & \multicolumn{1}{c|}{2.29}           & \multicolumn{1}{c|}{6.67}                                                          & \multicolumn{1}{c|}{7.97}                                                         & 9.20                                                                \\ \cline{2-14} 
\multicolumn{1}{|c|}{}                                                                                     & \multicolumn{1}{l|}{\textbf{Riding}}                                                         & \multicolumn{1}{l|}{$\overline{t_t}$} & min                                                & \multicolumn{1}{c|}{15.37}          & \multicolumn{1}{c|}{16.66}          & \multicolumn{1}{c|}{16.17}                                                         & \multicolumn{1}{c|}{15.06}                                                        & 14.27                                                               & \multicolumn{1}{c|}{15.37}          & \multicolumn{1}{c|}{16.66}          & \multicolumn{1}{c|}{14.56}                                                         & \multicolumn{1}{c|}{13.88}                                                        & 13.49                                                               \\ \hline
\multicolumn{1}{|c|}{\multirow{3}{*}{\textbf{\begin{tabular}[c]{@{}c@{}}Time\\ std. dev.\end{tabular}}}}   & \multicolumn{1}{l|}{\textbf{Access}}                                                         & \multicolumn{1}{l|}{$\sigma_{t,a}$}   & min                                                & \multicolumn{1}{c|}{0.00}           & \multicolumn{1}{c|}{0.00}           & \multicolumn{1}{c|}{2.12}                                                          & \multicolumn{1}{c|}{2.60}                                                         & 2.92                                                                & \multicolumn{1}{c|}{0.00}           & \multicolumn{1}{c|}{0.00}           & \multicolumn{1}{c|}{2.79}                                                          & \multicolumn{1}{c|}{3.09}                                                         & 3.31                                                                \\ \cline{2-14} 
\multicolumn{1}{|c|}{}                                                                                     & \multicolumn{1}{l|}{\textbf{Waiting}}                                                        & \multicolumn{1}{l|}{$\sigma_{t,w}$}   & min                                                & \multicolumn{1}{c|}{0.80}           & \multicolumn{1}{c|}{1.32}           & \multicolumn{1}{c|}{2.87}                                                          & \multicolumn{1}{c|}{3.52}                                                         & 4.19                                                                & \multicolumn{1}{c|}{0.80}           & \multicolumn{1}{c|}{1.32}           & \multicolumn{1}{c|}{3.85}                                                          & \multicolumn{1}{c|}{4.60}                                                         & 5.31                                                                \\ \cline{2-14} 
\multicolumn{1}{|c|}{}                                                                                     & \multicolumn{1}{l|}{\textbf{Riding}}                                                         & \multicolumn{1}{l|}{$\sigma_{t,r}$}   & min                                                & \multicolumn{1}{c|}{9.63}           & \multicolumn{1}{c|}{10.42}          & \multicolumn{1}{c|}{10.41}                                                         & \multicolumn{1}{c|}{9.73}                                                         & 9.10                                                                & \multicolumn{1}{c|}{9.63}           & \multicolumn{1}{c|}{10.42}          & \multicolumn{1}{c|}{9.35}                                                          & \multicolumn{1}{c|}{8.69}                                                         & 8.13                                                                \\ \hline
\multicolumn{1}{|c|}{\multirow{8}{*}{\textbf{Cost}}}                                                       & \multicolumn{1}{l|}{\textbf{Access}}                                                         & \multicolumn{1}{l|}{$c_a$}            & \$                                                 & \multicolumn{1}{c|}{0.00}           & \multicolumn{1}{c|}{0.00}           & \multicolumn{1}{c|}{111.00}                                                        & \multicolumn{1}{c|}{167.00}                                                       & 210.88                                                              & \multicolumn{1}{c|}{0.00}           & \multicolumn{1}{c|}{0.00}           & \multicolumn{1}{c|}{193.08}                                                        & \multicolumn{1}{c|}{236.06}                                                       & 271.79                                                              \\ \cline{2-14} 
\multicolumn{1}{|c|}{}                                                                                     & \multicolumn{1}{l|}{\textbf{Waiting}}                                                        & \multicolumn{1}{l|}{$c_w$}            & \$                                                 & \multicolumn{1}{c|}{45.70}          & \multicolumn{1}{c|}{75.64}          & \multicolumn{1}{c|}{163.88}                                                        & \multicolumn{1}{c|}{201.48}                                                       & 239.28                                                              & \multicolumn{1}{c|}{45.70}          & \multicolumn{1}{c|}{75.64}          & \multicolumn{1}{c|}{220.21}                                                        & \multicolumn{1}{c|}{263.05}                                                       & 303.48                                                              \\ \cline{2-14} 
\multicolumn{1}{|c|}{}                                                                                     & \multicolumn{1}{l|}{\textbf{\begin{tabular}[c]{@{}l@{}}Riding\\ (x-dir.)\end{tabular}}}      & \multicolumn{1}{l|}{$c_{t,x}$}        & \$                                                 & \multicolumn{1}{c|}{294.80}         & \multicolumn{1}{c|}{294.80}         & \multicolumn{1}{c|}{294.80}                                                        & \multicolumn{1}{c|}{294.80}                                                       & 294.80                                                              & \multicolumn{1}{c|}{294.80}         & \multicolumn{1}{c|}{294.80}         & \multicolumn{1}{c|}{294.80}                                                        & \multicolumn{1}{c|}{294.80}                                                       & 294.80                                                              \\ \cline{2-14} 
\multicolumn{1}{|c|}{}                                                                                     & \multicolumn{1}{l|}{\textbf{\begin{tabular}[c]{@{}l@{}}Riding\\ (y-dir.)\end{tabular}}}      & \multicolumn{1}{l|}{$c_{t,y}$}        & \$                                                 & \multicolumn{1}{c|}{43.33}          & \multicolumn{1}{c|}{71.72}          & \multicolumn{1}{c|}{60.94}                                                         & \multicolumn{1}{c|}{36.60}                                                        & 19.07                                                               & \multicolumn{1}{c|}{43.33}          & \multicolumn{1}{c|}{71.72}          & \multicolumn{1}{c|}{25.56}                                                         & \multicolumn{1}{c|}{10.50}                                                        & 2.07                                                                \\ \cline{2-14} 
\multicolumn{1}{|c|}{}                                                                                     & \multicolumn{1}{l|}{\textbf{\begin{tabular}[c]{@{}l@{}}Operational\\ (x-dir.)\end{tabular}}} & \multicolumn{1}{l|}{$c_{o,x}$}        & \$                                                 & \multicolumn{1}{c|}{359.23}         & \multicolumn{1}{c|}{223.44}         & \multicolumn{1}{c|}{112.33}                                                        & \multicolumn{1}{c|}{98.85}                                                        & 98.69                                                               & \multicolumn{1}{c|}{538.84}         & \multicolumn{1}{c|}{335.17}         & \multicolumn{1}{c|}{125.39}                                                        & \multicolumn{1}{c|}{113.58}                                                       & 116.71                                                              \\ \cline{2-14} 
\multicolumn{1}{|c|}{}                                                                                     & \multicolumn{1}{l|}{\textbf{\begin{tabular}[c]{@{}l@{}}Operational\\ (y-dir.)\end{tabular}}} & \multicolumn{1}{l|}{$c_{o,y}$}        & \$                                                 & \multicolumn{1}{c|}{52.80}          & \multicolumn{1}{c|}{54.36}          & \multicolumn{1}{c|}{37.08}                                                         & \multicolumn{1}{c|}{28.04}                                                        & 22.02                                                               & \multicolumn{1}{c|}{79.19}          & \multicolumn{1}{c|}{81.54}          & \multicolumn{1}{c|}{31.07}                                                         & \multicolumn{1}{c|}{19.72}                                                        & 9.67                                                                \\ \cline{2-14} 
\multicolumn{1}{|c|}{}                                                                                     & \multicolumn{1}{l|}{\textbf{Vehicle}}                                                        & \multicolumn{1}{l|}{$c_v$}            & \$                                                 & \multicolumn{1}{c|}{148.87}         & \multicolumn{1}{c|}{137.21}         & \multicolumn{1}{c|}{139.53}                                                        & \multicolumn{1}{c|}{167.99}                                                       & 185.61                                                              & \multicolumn{1}{c|}{223.30}         & \multicolumn{1}{c|}{205.82}         & \multicolumn{1}{c|}{148.23}                                                        & \multicolumn{1}{c|}{180.19}                                                       & 200.23                                                              \\ \cline{2-14} 
\multicolumn{1}{|c|}{}                                                                                     & \multicolumn{1}{l|}{\textbf{Total}}                                                          & \multicolumn{1}{l|}{$c$}              & \$                                                 & \multicolumn{1}{c|}{944.71}         & \multicolumn{1}{c|}{649.66}         & \multicolumn{1}{c|}{919.55}                                                        & \multicolumn{1}{c|}{994.77}                                                       & 1070.36                                                             & \multicolumn{1}{c|}{1225.16}        & \multicolumn{1}{c|}{1064.68}        & \multicolumn{1}{c|}{1038.35}                                                       & \multicolumn{1}{c|}{1117.89}                                                      & 1198.76                                                             \\ \hline
\end{tabular}
\end{table}
\end{landscape}

\section{Case Study}\label{sec:case_study}
To investigate the applicability of the described semi-on-demand hybrid routes in transit feeders, this case study makes use of a real-world transit network and demand data in the Chicago metropolitan area. We analytically classify areas served by feeder routes to each railway station into fixed route and flexible route service areas with analytical formula in Section~\ref{sec:op_x_f_s}. While actual real-world feeder designs would benefit from case-by-case in-depth analysis and simulation, this case study aims to provide insights how transit feeders may emerge with SAVs and semi-on-demand concept.

\subsection{Data Preparation and Geospatial Processing}

Demand data are adapted from the activity-based demand model CT-RAMP from Chicago Metropolitan Agency for Planning \citep{parsons_brinckerhoff_activity-based_2011}, where trips originate and end at micro analysis zones (MAZs). Only trips assigned to trains and originating outside walking distance are included, which equates to approximately 78,000 hourly trips after aggregation into three peak hours daily. The existing train stations of the commuter rail, Metra, \citep{city_of_chicago_metra_2012} and rapid transit system, Chicago “L”, \citep{chicago_transit_authority_cta_2022-1} are imported into the Geographic Information System (GIS) model. 

To apply the proposed analytical model, which assumes a directional corridor structure, to the complex two-dimensional network, a multi-step geospatial processing framework is implemented using Python. The processing pipeline involves the following steps:

\begin{enumerate}
    \item \textbf{Station Catchment Partitioning:} First, a Voronoi diagram is constructed based on existing Metra and CTA rail stations to define the catchment area of each station. The demand points (MAZs) are then assigned to their nearest station.
    
    \item \textbf{Corridor Axis Definition:} To map the two-dimensional Voronoi polygons onto the linear model ($x$-axis) described in Section~\ref{sec:math_form}, a principal axis is defined for each station zone. The line between the station to the furthest MAZ by Euclidean distance within the zone is established as the route's $x$-axis (representing the route length $L_x$).
    
    \item \textbf{Demand Projection and Separation:} All other MAZs within the zone are mapped relative to the principal axis. MAZs are grouped into two subzones based on their locations relative to the y-axis (perpendicular to the principal axis). Each sub-zone is segmented into one or more directional service corridors.
    
    \item \textbf{Parameter Extraction:} The x- and y-axis distances of each MAZ are calculated with respect to the defined corridor. The demand density $f(x)$ is derived by aggregating trips from MAZs at their x-coordinates. The detour/access distance ($y$) is approximated based on the catchment width perpendicular to the axis at given $x$ intervals.
    
    \item \textbf{Classification:} Finally, the optimal cumulative demand threshold (the number of passengers to serve in flexible route portions) $F(x_f^*)$ derived in Eq.~\ref{eq:F_gen} is applied to the demand profile of each corridor to determine the cutoff point between the flexible and fixed service areas.
\end{enumerate}

We carry out the analysis in each feeder corridor with the optimal cumulative demand threshold, which suggests serving these passengers with flexible routes would not only reduce their travel costs, but also total costs of users and operators. Instead of a specific distribution assumption (e.g., uniform or triangular), this analysis directly captures $f(x)$ from the number of trips in each MAZ (at its center). 

The parameters are set as follows: the maximum access time is 15~min and the walking speed is 4~km/h, resulting in a catchment width of 2~km and average detour $\overline{d_{d,y}}$ of 0.67~km (assuming a uniform distribution of y-directional demand); other parameters follow the previous numerical example in Section~\ref{sec:num_eg_fleet}. The optimal hourly demand to be served within the flexible portion $F(x_f^*)$ is calculated as 35.5~pax/h with Eq.~\ref{eq:F_gen}, equivalent to around 9~pax/trip.

\subsection{Results and Analysis}

Figure~\ref{fig:voronoi_result} shows which MAZs are served with fixed routes (blue dots) and flexible routes (in green dots) in each zone (in gray boundary). The flexible route service areas are mostly found in suburbs and rural areas far from downtown, with sparser demand and bigger gaps between transit lines.

\begin{figure}[htbp]
  \centering
  \includegraphics[width=5.5in]{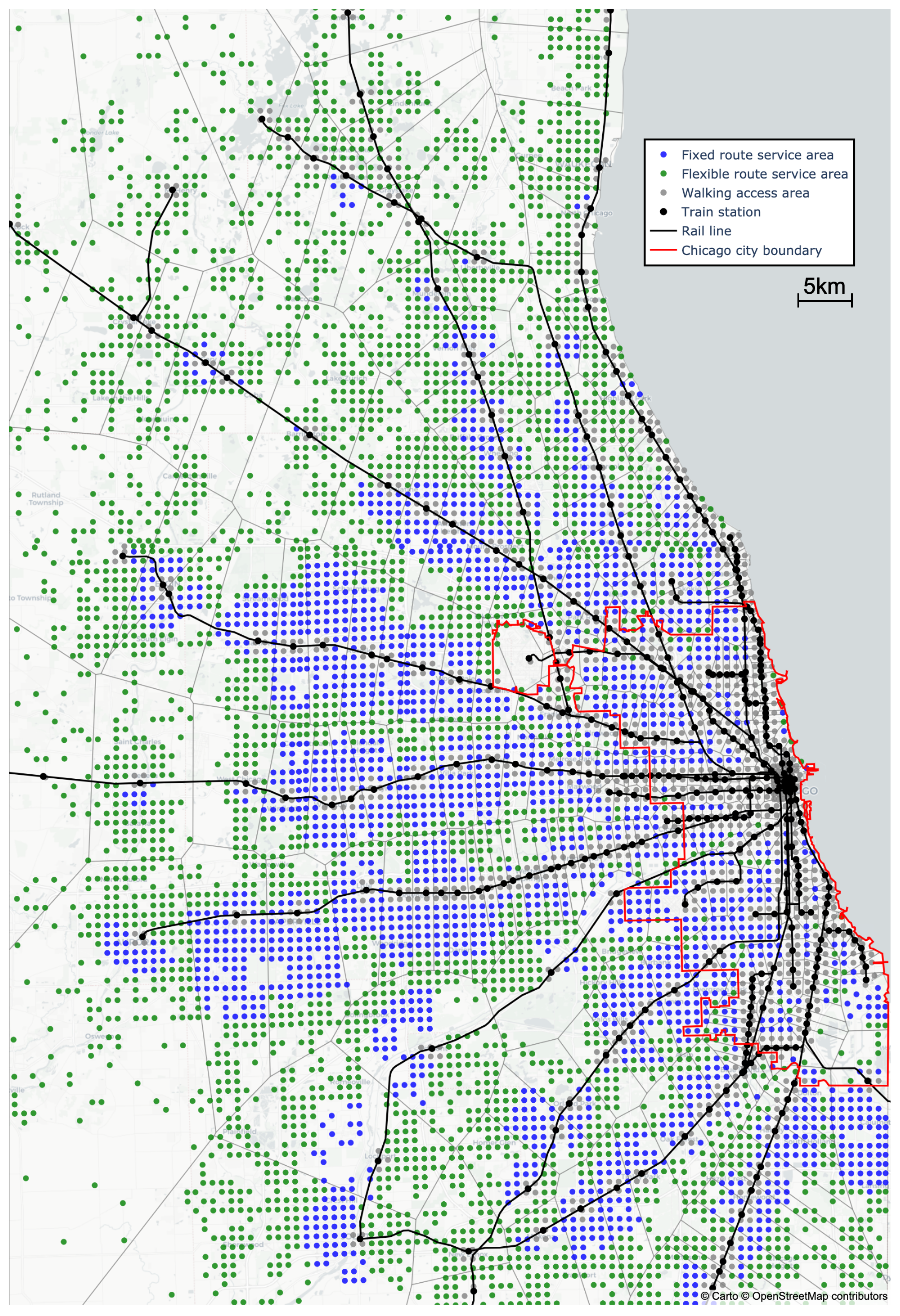}
  \caption{Case Study Results of Fixed Route and Flexible Route Service Area for Feeders to Train Stations \citep{chicago_transit_authority_cta_2022-1, city_of_chicago_metra_2012}}
  \label{fig:voronoi_result}
\end{figure}

\begin{landscape}
\begin{table}[]
\caption{Case Study Results of Fixed Route and Flexible Route Service Area for Feeders to Train Stations}
\label{tab:feeder_result}
\begin{tabular}{|ll|ccc|ccc|ccc|}
\hline
\multicolumn{2}{|l|}{\textbf{Service area}}                                                                        & \multicolumn{3}{c|}{\textbf{Among all feeders}}                                                                                                                                                                                                                   & \multicolumn{3}{c|}{\textbf{\begin{tabular}[c]{@{}c@{}}Among semi-on-demand \\hybrid route feeders\end{tabular}}}                                                                                                                                           & \multicolumn{3}{c|}{\textbf{\begin{tabular}[c]{@{}c@{}}Among MAZs in flexible\\ route service areas\end{tabular}}}                                                                                                                                                \\ \hline
\multicolumn{2}{|l|}{\textbf{Service mode}}                                                                        & \multicolumn{1}{c|}{\textbf{\begin{tabular}[c]{@{}c@{}}Fixed\\ route\end{tabular}}} & \multicolumn{1}{c|}{\textbf{\begin{tabular}[c]{@{}c@{}}Semi-on-\\ demand\\ route\end{tabular}}} & \textbf{\begin{tabular}[c]{@{}c@{}}Percent-\\ age \\ change\end{tabular}} & \multicolumn{1}{c|}{\textbf{\begin{tabular}[c]{@{}c@{}}Fixed\\ route\end{tabular}}} & \multicolumn{1}{c|}{\textbf{\begin{tabular}[c]{@{}c@{}}Semi-on-\\ demand\\ route\end{tabular}}} & \textbf{\begin{tabular}[c]{@{}c@{}}Percent-\\ age \\ change\end{tabular}} & \multicolumn{1}{c|}{\textbf{\begin{tabular}[c]{@{}c@{}}Fixed\\ route\end{tabular}}} & \multicolumn{1}{c|}{\textbf{\begin{tabular}[c]{@{}c@{}}Semi-on-\\ demand\\ route\end{tabular}}} & \textbf{\begin{tabular}[c]{@{}c@{}}Percent-\\ age \\ change\end{tabular}} \\ \hline
\multicolumn{2}{|l|}{\textbf{Number of feeder routes}}                                                             & \multicolumn{2}{c|}{909}                                                                                                                                                              & N/A                                                                       & \multicolumn{2}{c|}{787}                                                                                                                                                              & N/A                                                                       & \multicolumn{2}{c|}{787}                                                                                                                                                              & N/A                                                                       \\ \hline
\multicolumn{2}{|l|}{\textbf{Number of MAZs covered}}                                                              & \multicolumn{2}{c|}{5,732}                                                                                                                                                            & N/A                                                                       & \multicolumn{2}{c|}{5,479}                                                                                                                                                            & N/A                                                                       & \multicolumn{2}{c|}{3,562}                                                                                                                                                            & N/A                                                                       \\ \hline
\multicolumn{1}{|l|}{\textbf{\begin{tabular}[c]{@{}l@{}}Number of\\ passengers\end{tabular}}}     & \textbf{pax/h} & \multicolumn{2}{c|}{78,237}                                                                                                                                                           & N/A                                                                       & \multicolumn{2}{c|}{49,921}                                                                                                                                                           & N/A                                                                       & \multicolumn{2}{c|}{18,989}                                                                                                                                                           & N/A                                                                       \\ \hline
\multicolumn{1}{|l|}{\textbf{Average access time}}                                                & \textbf{min}   & \multicolumn{1}{c|}{7.50}                                                           & \multicolumn{1}{c|}{5.68}                                                                       & -24.3\%                                                                  & \multicolumn{1}{c|}{7.50}                                                           & \multicolumn{1}{c|}{4.65}                                                                       & -38.0\%                                                                  & \multicolumn{1}{c|}{7.50}                                                           & \multicolumn{1}{c|}{0.00}                                                                       & -100.0\%                                                                  \\ \hline
\multicolumn{1}{|l|}{\textbf{Average waiting time}}                                               & \textbf{min}   & \multicolumn{1}{c|}{7.50}                                                           & \multicolumn{1}{c|}{7.50}                                                                       & 0.0\%                                                                     & \multicolumn{1}{c|}{7.50}                                                           & \multicolumn{1}{c|}{7.50}                                                                       & 0.0\%                                                                     & \multicolumn{1}{c|}{7.50}                                                           & \multicolumn{1}{c|}{7.50}                                                                       & 0.0\%                                                                     \\ \hline
\multicolumn{1}{|l|}{\textbf{Average riding time}}                                                & \textbf{min}   & \multicolumn{1}{c|}{4.95}                                                           & \multicolumn{1}{c|}{6.12}                                                                       & 23.6\%                                                                    & \multicolumn{1}{c|}{5.96}                                                           & \multicolumn{1}{c|}{7.79}                                                                       & 30.8\%                                                                    & \multicolumn{1}{c|}{7.93}                                                           & \multicolumn{1}{c|}{12.75}                                                                      & 60.8\%                                                                    \\ \hline
\multicolumn{1}{|l|}{\textbf{Average user cost}}                                                & \textbf{\$}    & \multicolumn{1}{c|}{8.58}                                                           & \multicolumn{1}{c|}{7.90}                                                                       & -7.9\%                                                                    & \multicolumn{1}{c|}{8.86}                                                           & \multicolumn{1}{c|}{7.79}                                                                       & -12.0\%                                                                   & \multicolumn{1}{c|}{9.40}                                                           & \multicolumn{1}{c|}{6.60}                                                                       & -29.8\%                                                                   \\ \hline
\multicolumn{1}{|l|}{\textbf{\begin{tabular}[c]{@{}l@{}}Average operator\\ cost\end{tabular}}}  & \textbf{\$}    & \multicolumn{1}{c|}{0.32}                                                           & \multicolumn{1}{c|}{0.47}                                                                       & 45.6\%                                                                    & \multicolumn{1}{c|}{0.47}                                                           & \multicolumn{1}{c|}{0.69}                                                                       & 49.0\%                                                                    & \multicolumn{1}{c|}{N/A}                                                            & \multicolumn{1}{c|}{N/A}                                                                        & N/A                                                                       \\ \hline
\multicolumn{1}{|l|}{\textbf{\begin{tabular}[c]{@{}l@{}}Average\\ generalized cost\end{tabular}}} & \textbf{\$}    & \multicolumn{1}{c|}{8.90}                                                           & \multicolumn{1}{c|}{8.37}                                                                       & -6.0\%                                                                    & \multicolumn{1}{c|}{9.32}                                                           & \multicolumn{1}{c|}{8.49}                                                                       & -9.0\%                                                                    & \multicolumn{1}{c|}{N/A}                                                            & \multicolumn{1}{c|}{N/A}                                                                        & N/A                                                                       \\ \hline
\multicolumn{1}{|l|}{\textbf{Total access cost}}                                                  & \textbf{\$}    & \multicolumn{1}{c|}{322,731}                                                        & \multicolumn{1}{c|}{244,401}                                                                    & -24.3\%                                                                   & \multicolumn{1}{c|}{205,924}                                                        & \multicolumn{1}{c|}{127,594}                                                                    & -38.0\%                                                                   & \multicolumn{1}{c|}{78,330}                                                         & \multicolumn{1}{c|}{0}                                                                          & -100.0\%                                                                  \\ \hline
\multicolumn{1}{|l|}{\textbf{Total waiting cost}}                                                 & \textbf{\$}    & \multicolumn{1}{c|}{242,048}                                                        & \multicolumn{1}{c|}{242,048}                                                                    & 0.0\%                                                                     & \multicolumn{1}{c|}{154,443}                                                        & \multicolumn{1}{c|}{154,443}                                                                    & 0.0\%                                                                     & \multicolumn{1}{c|}{58,747}                                                         & \multicolumn{1}{c|}{58,747}                                                                     & 0.0\%                                                                     \\ \hline
\multicolumn{1}{|l|}{\textbf{Total riding cost}}                                                  & \textbf{\$}    & \multicolumn{1}{c|}{106,449}                                                        & \multicolumn{1}{c|}{131,619}                                                                    & 23.6\%                                                                    & \multicolumn{1}{c|}{81,780}                                                         & \multicolumn{1}{c|}{106,950}                                                                    & 30.8\%                                                                    & \multicolumn{1}{c|}{41,404}                                                         & \multicolumn{1}{c|}{66,574}                                                                     & 60.8\%                                                                    \\ \hline
\multicolumn{1}{|l|}{\textbf{Total user cost}}                                                  & \textbf{\$}    & \multicolumn{1}{c|}{671,228}                                                        & \multicolumn{1}{c|}{618,068}                                                                    & -7.9\%                                                                    & \multicolumn{1}{c|}{442,147}                                                        & \multicolumn{1}{c|}{388,987}                                                                    & -12.0\%                                                                   & \multicolumn{1}{c|}{178,481}                                                        & \multicolumn{1}{c|}{125,321}                                                                    & -29.8\%                                                                   \\ \hline
\multicolumn{1}{|l|}{\textbf{Total operating cost}}                                               & \textbf{\$}    & \multicolumn{1}{c|}{9,849}                                                          & \multicolumn{1}{c|}{16,179}                                                                     & 64.3\%                                                                    & \multicolumn{1}{c|}{9,414}                                                          & \multicolumn{1}{c|}{15,744}                                                                     & 67.2\%                                                                    & \multicolumn{1}{c|}{N/A}                                                            & \multicolumn{1}{c|}{N/A}                                                                        & N/A                                                                       \\ \hline
\multicolumn{1}{|l|}{\textbf{Total vehicle cost}}                                                 & \textbf{\$}    & \multicolumn{1}{c|}{15,151}                                                         & \multicolumn{1}{c|}{20,215}                                                                     & 33.4\%                                                                    & \multicolumn{1}{c|}{13,827}                                                         & \multicolumn{1}{c|}{18,891}                                                                     & 36.6\%                                                                    & \multicolumn{1}{c|}{N/A}                                                            & \multicolumn{1}{c|}{N/A}                                                                        & N/A                                                                       \\ \hline
\multicolumn{1}{|l|}{\textbf{Total operator cost}}                                              & \textbf{\$}    & \multicolumn{1}{c|}{25,000}                                                         & \multicolumn{1}{c|}{36,394}                                                                     & 45.6\%                                                                    & \multicolumn{1}{c|}{23,241}                                                         & \multicolumn{1}{c|}{34,635}                                                                     & 49.0\%                                                                    & \multicolumn{1}{c|}{N/A}                                                            & \multicolumn{1}{c|}{N/A}                                                                        & N/A                                                                       \\ \hline
\multicolumn{1}{|l|}{\textbf{\begin{tabular}[c]{@{}l@{}}Total generalized\\ cost\end{tabular}}}   & \textbf{\$}    & \multicolumn{1}{c|}{696,228}                                                        & \multicolumn{1}{c|}{654,462}                                                                    & -6.0\%                                                                    & \multicolumn{1}{c|}{465,388}                                                        & \multicolumn{1}{c|}{423,622}                                                                    & -9.0\%                                                                    & \multicolumn{1}{c|}{N/A}                                                            & \multicolumn{1}{c|}{N/A}                                                                        & N/A                                                                       \\ \hline
\end{tabular}
\end{table}
\end{landscape}

The detailed results are shown in Table~\ref{tab:feeder_result}. 787 routes (86.6\% of the total) are semi-on-demand hybrid routes, serving 3,562 MAZs (62.1\%) with flexible routes. However, this accounts for only 18,989~pax/h (24.3\%). The large coverage area accounting for a small portion of demand highlights the use case of flexible feeders for low-demand-density areas. Implementing these hybrid routes leads to an estimated 49.0\% increase in operator costs, primarily due to increased vehicle mileage from detours and potentially larger fleet requirements. This, however, is a small portion of the total costs, based on the SAV cost forecast. Higher operating and vehicle costs would reduce $x_f^*$ and the extent of flexible routes. 

While hybrid routes only save 7.9\% user cost and 6.0\% total cost in general, they make significant differences among travelers who use the flexible route portion by eliminating the access cost, accounting for a saving of 29.8\% user cost. This enhances the attractiveness of transit systems to travelers in these zones who are usually further from the train stations.\footnote{The benefits brought by mode shift to transit are however not captured in this study.}

\section{Conclusion}\label{sec:concl}

\subsection{Summary}

This study considers semi-on-demand hybrid route service in public transit systems. It serves directional demand by first offering passengers further from a transit station/downtown with on-demand flexible route service and then continuing with a traditional fixed route. It combines the economies of scale of fixed route bus service and the accessibility and flexibility of taxi and shared autonomous vehicles (SAVs).

We develop an analytical approach to delineate the conditions in which each of the route forms (fixed, hybrid, and flexible) is optimal, with tractable cost expressions that consider the total costs of users (access, waiting, and riding) and operators (operating and vehicle) in two formulations to support respectively strategic and tactical decisions. Closed-form expressions are derived to determine the optimal flexible route portion and fleet size for hybrid routes, and consider the optimal headway with variable vehicle sizes. Through numerical examples and a case study in the Chicago metropolitan area, we demonstrate the benefits and applications of semi-on-demand feeders. 

The findings provide insights into the choice between fixed and flexible route service. While the model is technology-neutral, the lower operating cost brought by SAVs strongly favor more flexible route service. We also demonstrate the benefits of hybrid routes in transit networks. Serving passengers located further away with flexible routes lowers total costs, particularly user costs in access, which could attract more riders to connect with the main transit system. Besides, demand gradients favor longer flexible routes, a good fit for cities with urban sprawl. 

In the example of the joint design of flexible route portion, headway, fleet size, and vehicle size, vans operating fully flexible routes emerged as the system-optimal solution under forecasted SAV costs. However, minibuses with hybrid routes still play a key role when SAV costs are still high and the operating budget is limited, or during transitions to smaller vehicles and flexible routes to introduce the on-demand service and enhance service attractiveness. These are significant insights for transit agencies to plan investments in terms of fleet and vehicle sizes for their network at different levels of SAV operating costs. Our case study also identifies flexible route areas in suburbs and rural areas further from the station and with coverage gaps, helping to serve more riders.

\subsection{Limitations and Future Research Directions}

This study develops an analysis tool for semi-on-demand hybrid route services with tractable and closed-form solutions. The resulting analytical formulation of cost and optimal flexible route portion and fleet size can aid transit agencies in planning new routes and making vehicle investment decisions in the era of SAVs. It provides a continuous approximation alternative to computationally intensive simulation-based optimization for research in large-scale transit network design with SAVs. While yielding practical results for easy adoption by policymakers and researchers in a more sophisticated simulation-optimization framework, it is inherently limited by several simplifying assumptions, including directional and inelastic demand, a binary fixed/flexible service model along the route, no vehicle backtracking, and idealized headway adherence. The formulation allows future refinement and investigation in the following areas.

First, it can be combined with demand- and supply-side models in the transit network design problem. Demand models to investigate mode shift towards the transit system with improved first-mile-last-mile connectivity, and its competition with driving. For the supply side, a comprehensive joint transit network design model could leverage the analytical expressions derived here to evaluate network-wide implications of SAV integration, considering inter-route synergies and competition. While this model assumes a static demand, it can be readily applied to problems with multiple horizons and time-varying design parameters and demand. The methodology can also be generalized to other network structures by transforming the space with the x-axis as the shortest path and the y-axis as the detours.

Second, research can focus on the study and forecast of connected and automated vehicles, including cost parameters. Our numerical example and sensitivity study illustrate the attractiveness of vehicles larger than sedans to provide hybrid route feeder services, and vehicle size matters in hybrid route design and headway-fleet size decisions. Following the analytical approach developed in this paper, more accurate cost forecasts would facilitate new service planning and deployment.

Third, research with agent-based simulation models and fleet control can verify the model results at a microscopic level and assess the impacts of several assumptions and design parameters. For example, the formulation assumes sequential drop-offs before pick-ups, constant stop spacing for the fixed route portion, and at most one flexible and fixed route portion. It can be extended to a more complex optimization problem with endogenous stop spacing optimization and multiple flexible route portions. Backtracking is not allowed in the model, so solving the vehicle routing problem that optimizes pick-ups/drop-offs could improve the performance of the flexible route portion. A constant headway is assumed in each case, whereas simulations that account for the effects of demand and service variances on waiting time could assess schedule adherence. Further demand elasticity consideration, such as with walking distance, would better capture riders' choice between flexible and fixed route services. Vehicle capacity buffer allows for stochastic demand, with room for future statistical or simulation approaches to evaluate the effects of vehicle size and demand fluctuation on boarding rejection. Lastly, door-to-door service assumed within the flexible route portion may be expanded to incorporate the concept of meeting points that would likely shift the optimal solution towards larger flexible route portions by reducing the marginal cost of detours.

\section*{Final Remark}
The authors dedicate this paper to the memory of our co-author, friend, and mentor, Professor Hani S. Mahmassani, whose guidance was instrumental in the conception and execution of this research.

\section*{Acknowledgments}
The authors would like to thank Antonios Tsakarestos of the Technical University of Munich for the valuable comments in the initial stage of the study.

\section*{Author Contributions}
Conceptualization: MN, HSM;  Formal analysis: MN, FD, HSM, KB; Methodology: MN, FD, HSM; Supervision: HSM, KB; Visualization: MN; Writing – original draft: MN, FD; Writing – review and editing: HSM, KB

\section*{Funding Sources}
This research did not receive any specific grant from funding agencies in the public, commercial, or not-for-profit sectors.


\bibliographystyle{elsarticle-harv} 
\bibliography{references} 


\end{document}